\documentclass{imsart}

\usepackage{amsmath,amsfonts,amssymb,amsthm}
\usepackage{graphicx}
\usepackage[round]{natbib}
\usepackage{subcaption}
\usepackage{mwe}
\usepackage{paralist} 
\usepackage{kbordermatrix}
\usepackage{verbatim}
\usepackage{booktabs}
\usepackage{indentfirst} 
\usepackage{xcolor}

\usepackage[font=small,labelfont=bf]{caption}
\usepackage[normalem]{ulem}

\DeclareMathOperator*{\argmin}{arg\,min}

\renewcommand{\le}{\leqslant}
\renewcommand{\ge}{\geqslant}
\renewcommand{\leq}{\leqslant}
\renewcommand{\geq}{\geqslant}

\newcommand{\bsx}{\boldsymbol{x}}

\newcommand{\cixzk}{\ci}

\newcommand{\dnorm}{\mathcal{N}}

\newcommand{\diag}{\mathrm{diag}}

\newcommand{\eff}{\mathrm{Eff}}
\newcommand{\bc}{\mathrm{BC}}
\newcommand{\ts}{\mathrm{TS}}
\newcommand{\ik}{\mathrm{IK}}

\newcommand{\rd}{\,\mathrm{d}}

\newcommand{\e}{\mathbb{E}}
\newcommand{\var}{\mathrm{var}}

\newcommand{\dunif}{\mathbb{U}}
\newcommand{\real}{\mathbb{R}}

\newcommand{\tran}{\mathsf{T}}

\newcommand{\ci}{\mathcal{I}}

\newcommand{\cw}{\mathcal{W}}
\newcommand{\cx}{\mathcal{X}}
\newcommand{\cy}{\mathcal{Y}}

\newcommand{\giv}{\!\mid\!} 
\newcommand{\err}{\varepsilon}
\newcommand{\phe}{\phantom{=}}
\newcommand{\xiVal}{(2i-N-1)/N}

\newcommand{\thresh}{\mathrm{thresh}}

\newtheorem{theorem}{Theorem}
\newtheorem{lemma}{Lemma}

\newtheorem{proposition}{Proposition}

\theoremstyle{definition}

\begin{document}

\begin{frontmatter}

 \title{Kernel regression analysis of tie-breaker designs \thanks{This version of the paper is published in the \textit{Electronic Journal of Statistics} with open access (DOI: doi.org/10.1214/23-EJS2102). There are formatting differences between the arXiv and journal versions.}}
 \runtitle{Kernel regression analysis of tie-breaker designs}

 \begin{aug}
\author{\fnms{Dan M.} \snm{Kluger}\ead[label=e1]{kluger@stanford.edu}}
\and
\author{\fnms{Art B.} \snm{Owen}\ead[label=e2]{owen@stanford.edu}}

\address{Department of Statistics, Stanford University, Stanford CA, 94305 \\
\printead{e1,e2}}

 \end{aug}

\runauthor{Kluger and Owen}

\maketitle

\begin{abstract}
Tie-breaker experimental designs are hybrids of Randomized Controlled Trials (RCTs) and Regression Discontinuity Designs (RDDs) in which subjects with moderate scores are placed in an RCT while subjects with extreme scores are deterministically assigned to the treatment or control group. In settings where it is unfair or uneconomical to deny the treatment to the more deserving recipients, the tie-breaker design (TBD) trades off the practical advantages of the RDD with the statistical advantages of the RCT. The practical costs of the randomization in TBDs can be hard to quantify in generality, while the statistical benefits conferred by randomization in TBDs have only been studied under linear and quadratic models. In this paper, we discuss and quantify the statistical benefits of TBDs without using parametric modelling assumptions. If the goal is estimation of the average treatment effect or the treatment effect at more than one score value, the statistical benefits of using a TBD over an RDD are apparent. If the goal is nonparametric estimation of the mean treatment effect at merely one score value, we prove that about 2.8 times more subjects are needed for an RDD in order to achieve the same asymptotic mean squared error.  We further demonstrate using both theoretical results and simulations from the \cite{AngristLavy} classroom size dataset, that larger experimental radii choices for the TBD lead to greater statistical efficiency.

\end{abstract}


\tableofcontents

\end{frontmatter}

\section{Introduction} \label{sec:Intro}

In this paper we study a nonparametric regression approach
to tie-breaker studies. 
In the settings of tie-breaker studies, there is a costly
treatment while the control is inexpensive or even free. In addition, an investigator can decide how to allocate the costly treatment using a priority ordering on the subjects. The priority ordering could be based on how deserving of the treatment each subject is, or based on how strongly each subject is expected to respond to the treatment.  
Examples include offering
scholarships or school placement to students \citep{Angrist2020Nebraska}, offering 
a drug rehabilitation program to people
of varying needs \citep{CappelleriXX},
or assigning interventions to reduce
risk factors for child abuse and neglect \citep{kran:2022}.
In these settings a randomized controlled trial
(RCT) is inappropriate because it is extremely
inefficient economically, or even ethically
questionable.
The natural, even automatic, approach to settings like
this is to rank the subjects $i=1,\dots,N$ 
according to their value of a running variable $x_i$
and assign the treatment to only those subjects with
the highest values of $x_i$. For simplicity one can
assume that the number of subjects to treat is
fixed and that the treatment is offered to subject
$i$ if and only if $x_i> t$ for some threshold $t$.

The problem with a deterministic treatment based on
$x_i$ is that it complicates causal inference of the
effect of the treatment.  One can use regression discontinuity
analysis \citep{this:camp:1960,cattaneo2022regression} but
the regression discontinuity design (RDD) is known to 
give treatment effect estimates a very high variance.
See, for example, \cite{Gelman_dont_do_high_order_polynomial}.
In a parametric regression model the treatment and running
variable are highly correlated, making for an inefficient
design \citep{gold:1972,jacob2012practical}.
In a tie-breaker design (TBD), the subjects are given the
treatment with a probability that increases with $x_i$.
The top ranked subjects get the treatment with probability one,
the bottom ranked subjects do not get the treatment and subjects
in between are randomly assigned to either treatment or control.
The tie-breaker design interpolates between two extremes: the RCT
and the RDD, trading off statistical efficiency with the short term economic value 
of aligning treatment to the running variable. 

We believe that there are many more good uses for TBDs.  Many companies interact electronically
with their customers and partners.
Perks, such as service upgrades, can
easily be assigned with some randomization. Because the treatment
is costly, it is important to
evaluate the treatment efficacy later. 
This provides a strong motivation to introduce
randomization. 
When the perk is simply a gift to some subjects there
is less ethical concern over whether it goes to the
most loyal customers or introduces some randomization.
We also expect that tie-breakers will be useful
in evaluating governmental programs such as
the one in \cite{kran:2022}, as well as educational programs such as those in \cite{Angrist2020Nebraska}.

TBDs have been primarily studied as experimental
designs using parametric regression modeling assumptions.  While the design literature
focuses on parametric models, the RDD literature primarily uses nonparametric
regression methods.  In this paper, we quantify the statistical gains to
be obtained by conducting a TBD instead of an RDD using nonparametric regression.

Our main theoretical contributions are as follows.
We study a kernel
weighted local linear regression with a slope
and intercept for both treated and control
subjects and bandwidth $h$. 
The RDD can consistently estimate
the treatment effect only at $x=t$, and so
we focus our comparison at that point. We find
an expression for the optimal bandwidth
for estimating the treatment effect at $x=t$ under the TBD.
We then compare the optimal mean squared
error at $x=t$ for the two designs. 
For the popular triangular kernel, a TBD
reduces the
asymptotic mean squared error (AMSE) by
a factor of about $2.27$ compared to an RDD
of the same sample size, $N$.  For other popular
kernels, the AMSE is reduced by a slightly greater
factor.  In this setting, the AMSE decreases
proportionally to $N^{-4/5}$, and using $N$
points in an RDD is comparable to using only
$0.36N$ points in a TBD.  
The asymptotic analysis has a bandwidth
that converges to zero.
Since this convergence
is at the very slow $N^{-1/5}$ rate, we cannot
assume that in practice $h$ will be small enough such that subjects without randomized treatment are discarded. Therefore, we also compare the designs when $h$ is fixed and large enough to include nonnrandomized subjects in the regression.  In this setting, the
efficiency ratio, which we define as the relative variance of treatment effect estimators under the two designs, can be as large as four. For a fixed bandwidth and for the triangular and boxcar kernels, we also find that the efficiency ratio is monotone non-decreasing in the proportion of subjects who are given a randomized treatment assignment.

A further advantage of the TBD is that
it can give consistent nonparametric estimates
of the treatment effect for any value of $x$
in the randomization window.
It can also be used to estimate the average
causal effect over that window. These additional advantages are described more explicitly in Section \ref{sec:CausalParamsAndProblem}.

An outline of this paper is as follows. 
Section~\ref{sec:lit} reviews the literature on
tie-breakers as well as the much larger literature
on RDDs.  In Section~\ref{sec:CausalParamsAndProblem}, we define a causal parameter of interest that can be used to compare the TBD to the RDD. We also introduce the causal identification assumptions needed and the local linear regressions used for estimating that parameter. In Section~\ref{sec:bandwidth_shrink_asymptotics}, we compare the mean squared error (MSE) in asymptotically optimal estimation of our causal parameter of interest under an RDD to that under a TBD. In that asymptotic setting, the optimal bandwidth decreases at the slow $O(N^{-1/5})$ rate and then the local linear regression is eventually supported entirely in the experimental region of the TBD. In Section~\ref{sec:AsymApprox_via_integrals}, we investigate another regime where the bandwidth $h$ is fixed and is assumed to be larger than the radius $\Delta$ of the experimental region. For this setting, deferring an investigation of the bias to Appendix \ref{sec:MagicBandwidth}, we study the variance of our estimator as a function of $\Delta$ and find the efficiency ratio to be monotone in $\Delta$ for the triangular and boxcar kernels.   Section~\ref{sec:Data_application}
shows how one can compute efficiency ratios empirically using one's
actual assignment variable levels, focusing on the Israeli classroom size data from \cite{AngristLavy} as an example.  The curves of the empirical efficiency ratios are quite similar to the ones obtained theoretically. Section~\ref{sec:Conclusion} presents a discussion. Appendix \ref{sec:GenPneqHalf} extends our results to TBDs in which each subject in the experimental group is given the treatment with probability $p \neq 1/2$. Appendices \ref{sec:proof_of_TBD_MSE}, \ref{sec:relMSETheoremProof}, \ref{sec:TSefficiencyCalc}, and \ref{sec:MonotoneEffTS} contain some of our proofs.

\section{Literature review}\label{sec:lit}

Here we survey the small TBD literature
and some recent developments in the
much larger RDD literature.
We also note connections to the experimental
design literature.
Most of the TBD literature has focused
on global parametric models.  Those are also
the dominant model for experimental design.
The TBD is usually compared to the RDD, for which nonparametric models are the norm.

\subsection{Regression discontinuity methods}

Here we present some concepts from the regression 
discontinuity literature drawing heavily on 
\cite{cattaneo2022regression}.
We begin with a setting where there is a running
variable $x_i$ (also called a score or index)
for subject $i=1,\dots,N$.  
Subjects with $x_i>t$ are given
the treatment and others get the control.
The treatment levels are typically $T_i\in\{0,1\}$
with $T_i=0$ being the control.  For the TBD
setting it is more convenient to use $Z_i\in\{-1,1\}$
with $Z_i=-1$ indicating the control. 
The potential outcomes for subject $i$
are $Y_{i+}$ if treated and $Y_{i-}$ for control.

There are two main approaches to RDD in causal 
inference, continuity-based and local
randomization-based.  The continuity-based approach
assumes that the mean response
for treated subjects is continuous in $x$
as is that for control subjects.  If the mean response
for all subjects shows a discontinuity at $x=t$, then
the magnitude of this discontinuity is defined to
be the causal effect of treatment on subjects at $x=t$,
and one then considers how to estimate that effect.
The version from \cite{hahn2001identification}
has IID tuples $(x_i,Y_{i+},Y_{i-})$ where
$Y_i$ equals $Y_{i+}$ for $x_i>t$ and $Y_i$
equals $Y_{i-}$ otherwise.
The treatment effect is then
$$
\tau = \lim_{x\downarrow t}\mu_+(x)
-\lim_{x\uparrow t}\mu_-(x)
$$
where $\mu_{\pm} = \e(Y_{\pm}\giv X=x)$.
We will work primarily 
with a superpopulation setting
where the subjects in the study are sampled
from a joint distribution. \cite{cattaneo2022regression} discuss this
setting along with some other 
settings that focus on causal inference
for the given subjects.

The local randomization approach from \cite{cattaneo2015randomization}
assumes the existence of a window $\cw=[t-h,t+h]$
such that for $x\in \cw$ the treatment variable is
`as good as randomized'. 
In the local randomization approach
we assume {\bf a)} that the joint distribution of 
the $Z_i$ for $x_i\in\cw$ is known, and,
{\bf b)} that the potential outcomes 
$(Y_{i+},Y_{i-})$ are independent of $x_i$.
In particular, both mean responses must
be constant functions of $x\in\cw$.
A variant of local randomization 
has the treatment based on 
a threshold of $x$ where $x$ is a noisy
version of a latent variable $u$ 
\citep{eckl:igna:wage:2020} and where we have
outside information under which the 
probability of treatment given $u$ is known.

Both frameworks have challenges.
The obvious difficulty with local randomization 
is choosing the window $\cw$ (or knowing
the treatment probability given $u$).  
A smaller window provides a smaller dataset to use while making the
window larger will normally increase 
the discrepancy between
the model and the ground truth. 
The TBD can be viewed as a strategy to impose
by design the first assumption in the local 
randomization approach while not making the 
second assumption.

The challenge in the continuity framework is
in estimating the necessary limits.  In a parametric
model, those limits are estimated from all the
data but have a bias due to lack of fit of the
parametric model.  As a result, nonparametric
regression methods based on local polynomial
models are favored.  
The challenge there is that one must choose
a bandwidth $h$, analogously to the window
size from the local randomization
framework.  Because the mean responses are only
locally polynomial we must contend with
a bias-variance tradeoff in estimating the limits.

Our theoretical and numerical results compare
the TBD to the RDD in the continuity framework.
We think that this is the more likely alternative
analysis for our motivating applications if 
randomization had not been used, because
the local randomization assumptions 
do not seem natural in those applications. 
We focus on the
accuracy of point estimation.  There is also a large
literature on constructing confidence intervals
around the RDD estimate (see \cite{cattaneo2022regression}).
We describe some of those concepts in this paper,
but we do not develop confidence intervals
for the TBD due to space constraints.

There are many different settings where treatments
depend in a discontinuous way on $x$.
In a sharp design, the treatment is $Z_i=1$
if and only if $x_i>t$.  
In a fuzzy design, the assignment to treatment
or control might not perfectly match
$x_i>t$ versus $x_i\le t$, for reasons
beyond the control of the investigator.
For instance there may be subjects that
do not comply with their assigned treatment.
A related issue is that some subjects might be
able to manipulate their value of the running
variable in order to get (or avoid) the treatment.
\cite{rosenman2019optimized} discuss some ways
to counter that problem.
In the settings we consider, the investigator
has control of the treatment and so we study
the sharp design.  We also do not address issues of subject compliance, as we suppose that the effect of the treatment assignment or the intent to treat is of sufficient interest to the investigator.

The RDD setting has been extended well beyond
the simple framework described above.
There are versions with treatments at more than
two levels as well as versions with 
continuous treatments.
The cutoff can be defined in terms of a vector
of covariates yielding a discontinuity set of
dimension one less than the vector has.
The treatment discontinuity could be defined
by geographical boundaries. 
There are multi-cutoff
settings where subject $i$ gets the treatment
when $x_i>t_i$.  There are models where it is
the derivative of $\e(Y\giv x,Z=1)-\e(Y\giv x,Z=-1)$ 
that has a step discontinuity at $t$.
For discussion and references to the variants
above, see \cite{cattaneo2022regression}.

\subsection{Experimental design}

While we study TBDs in comparison to
regression discontinuity,
they can also be considered within an
experimental design framework such as the
covariate-dependent designs considered by
\cite{metelkina2017information}. That paper
emphasizes sequential problems,
and like most of the design literature it
works primarily with parametric
models. 
They find conditions where the treatment
policies converge to optimal deterministic 
functions of a covariate vector.  
The TBD does not use
deterministic allocations which is an advantage
if the response distribution is subject to change
between experiments.

Experimental design, especially in a sequential
setting is closely related to bandit methods.
We are motivated by problems where the responses
$Y_i$ arrive too slowly for bandit methods to
be suitable.  In a business setting, the responses
may arrive after a year or calendar quarter while
the effect of a scholarship on graduation rates can
only be seen years later.

\subsection{Tie-breakers}

The simplest tie-breaker design replaces the threshold $t$ by
two thresholds $t\pm\Delta$.
Subjects with $x_i>t+\Delta$ get the treatment,
subjects with $x_i<t-\Delta$ get the control
and other subjects are randomized to 
either treatment or control.
The simplest choice has
\begin{equation}\label{eq:threelevel}
\Pr(Z_i=1\giv x_i)
=\begin{cases}
0, & x_i <t-\Delta,\\
\frac12, &|x_i-t|\le\Delta\\
1, & x_i >t+\Delta.
\end{cases}
\end{equation}
\cite{camp:1969} describes the $\Delta=0$ version of this design.
Some subjects are exactly at the 
threshold $t$ and then randomization breaks the 
ties among them. 
\cite{boru:1975} considers positive
values of $\Delta$ such that differences
in the running variable among subjects
with $|x-t|\le \Delta$ are essentially
arbitrary because $x$ is an imperfect measure.
\cite{abdu:etal:2022} study the New York
school system that
breaks ties among applicants by lottery, or
standardized test, or audition, depending on
the program.  We only consider randomized
tie-breaking.

\cite{gold:1972} considers a simple two line regression model that in our
notation is
\begin{align}\label{eq:ovmodel}
Y_i & = \beta_1 +\beta_2x_i+\beta_3Z_i+\beta_4x_iZ_i+\err_i
\end{align}
for IID errors $\err_i$ with mean zero 
and variance $\sigma^2$.
He finds that an RDD estimates these
coefficients with a variance that is 
asymptotically $\pi/(\pi-2)\approx2.75$ times as large
as it would be under an RCT.  His setting
has Gaussian $x_i$ and $\beta_4=0$. \cite{jacob2012practical} generalizes the above model to polynomials of
degree two or three in $x$ with or
without interactions between $x$
and $Z$. Their Table 6 shows that an RCT is 4 times as efficient as the RDD for
a uniformly distributed running variable
with $t$ at the midpoint of its range.
They also provide similar efficiency estimates
for other polynomial models for both uniform
and Gaussian $x$ and include settings
where $t$ is not at the median of the
distribution of $x$.

The above comparisons of RCTs to RDDs do not
include tie-breakers.  \cite{CDP94_TBD_power3Delta}
compare small, medium and large randomization
windows in which 20\%, 35\% and 50\%, respectively,
of the subjects get a randomized treatment.
They tabulate the sample sizes needed to attain
a certain level of statistical power for three treatment effect sizes
in these TBDs as well as in an RDD and in an RCT.  All
designs had half of the subjects getting the
treatment. The running variable $x$ was normally
distributed.  The model was~\eqref{eq:ovmodel}
with $\beta_4=0$, making the treatment effect
constant. The power calculations were done
by Monte Carlo sampling. The required sample
sizes became smaller with increased
randomization at any level of power and
effect size.

\cite{owen:vari:2020} work out the 
asymptotic variance of $\hat\beta$ in the
model~\eqref{eq:ovmodel} as a function of
$\Delta$.  They consider both $\dunif[-1,1]$
and $\dnorm(0,1)$ distributions for $x$ and
a threshold $t$ at the median of $x$'s distribution.
The estimated treatment effect is 
$2(\hat\beta_3+\hat\beta_4x)$ and they find
for uniform $x$ that this estimate has
asymptotic variance proportional to
$16(1+3x^2)/[1+3\Delta^2(2-\Delta^2)]$
where $\Delta=0$ describes the RDD and
$\Delta=1$ is the RCT. This decreases
monotonically in $\Delta$ while increasing
monotonically in $|x|$.

They also consider the opportunity cost
of experimentation compared to the RDD.
For $x_i\sim\dunif[-1,1]$ the expected
value of $\sum_{i=1}^NY_i$ is approximately 
$(\beta_1 + \beta_4(1-\Delta^2)/2) N$. If
larger $Y_i$ are better and $\beta_4>0$ then
the opportunity cost 
grows proportionally to $\beta_4\Delta^2$.
They discuss how one might trade off this
opportunity cost against statistical efficiency.

The TBD has so far been analyzed for simpler
methods than the RDD has.  This can be understood
by comparing their workflows.  In a TBD we 
measure $x_i$, then sample $Z_i$ and then
some time later observe $Y_i$. For an RDD we
usually get $(x_i,Z_i,Y_i)$ all at once.
The investigator planning a TBD only has $x_i$,
must decide how to assign the $Z_i$,
and may not know what model will be fit later,
and then chooses some specific model to design for.
When one studies the TBD theoretically, one does
not even have the $x_i$ and then it is natural to
assume a distribution for them.
The TBD is prospective while the
RDD is retrospective.

When vectors $\bsx_i$ of covariates
are available, \citet[Section 8]{owen:vari:2020}
describe how to investigate numerically
the efficiency of a TBD that fits a regression
model on some collection of features of $\bsx_i$
that interact with $Z_i$ where the treatment
window is based on a linear combination of $\bsx_i$.

\cite{TimArt22} study multiple regression for
a tie-breaker in a regression model
$Y_i = \bsx_i^\tran\beta + Z_i\bsx_i^\tran\gamma+\err_i$
with $\Pr(Z_i=1)=p_i\in[0,1]$.
They study a prospective $D$-optimality criterion 
that maximizes the determinant of $\e(\cx^\tran\cx)$,
where $\cx$ is the (random) design matrix built from
$\bsx_i$ and $Z_i$.  
For any known $\bsx_i$, the finite sample optimal $p_i$
can be computed by convex optimization.  
For as yet unobserved $\bsx_i$ the prospective $D$-optimality
criterion averages over both random $Z_i$ and random
$\bsx_i$ from an assumed distribution for $\bsx_i$.
For random $\bsx_i$, they study
a three level tie-breaker with running variable
$\bsx_i^\tran\eta$ and treatment probabilities $0$, $0.5$
and $1$.  

\cite{owen:vari:2020} consider replacing the
simple trichotomy~\eqref{eq:threelevel} by various
sliding scales where $\Pr(Z_i=1 \giv x_i )$ is a monotone
function of $x_i$. They find no advantage to such
alternatives when $x_i$ has a symmetric distribution
about $t$ and half the subjects are treated.
\cite{HarrisonArt22} revisit that problem for the
two line model and find optimal designs for
general $x_i$ distributions and general fractions
of treated subjects without assuming that half
of the subjects will be treated.
These optimal designs can greatly improve 
upon the design defined in \eqref{eq:threelevel}. 
They still have $\Pr(Z_i=1 \giv x_i)$ as piecewise 
constant functions of $x_i$.
If we impose monotonicity $\Pr(Z_i=1 \giv x_i )\ge\Pr(Z_{i'}=1 \giv x_{i'})$
whenever $x_i\ge x_{i'}$, then only two treatment
probability levels are needed.

A limitation of previous comparisons between RDDs and TBDs is that they all assume parametric regression
models for the response $Y_i$.  We compare them using local linear
regression. 
For simplicity, we restrict our attention to the three level version of the TBD in~\eqref{eq:threelevel}.

Our model assumes an additive error
on top of smooth functions of $x$ for the 
data points where $|x-t|>\Delta$.
When $h>\Delta$, the causal estimate we consider merges deterministic
and randomized treatment allocations and then
cannot be analyzed in a potential outcomes framework.
It is common in causal inference to ignore
such data.  For instance, a rule of thumb
in \cite{crump2009dealing} is to omit data
where the treatment probability is outside
$[0.1,0.9]$.
Asymptotically, $h<\Delta$ and then a potential
outcomes analysis is available. Otherwise, to stay within the potential outcomes framework, one
must choose between ignoring some data and
using the additive error model like we do.

\section{Causal estimand and problem formulation} \label{sec:CausalParamsAndProblem}

Throughout the text we will compare the TBD to the RDD. In our comparison, we will define $t$ to be the putative RDD threshold and $\Delta$ to be the experimental radius, and we consider allocation of the treatments to the $N$ subjects according to the 3-level tie-breaker design~\eqref{eq:threelevel}.

Next we discuss the estimands of interest.
For each subject we consider the assignment variable $X\in\real$, the treatment $Z\in\{-1,1\}$ and two potential outcomes: 
$Y_+=Y(Z=1)$ and $Y_-=Y(Z=-1)$. Defining \begin{equation}\label{eq:muPlusMinus_def}
    \mu_+(x) \equiv \e \big(Y_+ \giv X=x \big)
\quad\text{and}\quad
\mu_-(x) \equiv  \e \big(Y_- \giv X= x \big),
\end{equation} the treatment effect at $X=x$ is
$$\tau(x) =\e( Y_+ -Y_-\giv X=x) = \mu_+(x)-\mu_-(x).$$
If the investigator chooses an RDD with a  threshold at $t$, then under certain regularity conditions, the causal estimand $$\tau_{\thresh} \equiv 
\tau(t)$$ can be consistently estimated. In particular, we assume IID samples (or sufficiently weak dependence between the samples) and that:
\begin{compactenum}[(i)]
    \item The density $f(\cdot)$ of the assignment variable $X$ is continuous at $t$ with $f(t)>0$.
    \item The conditional mean functions $\mu_+$ and $\mu_-$ in \eqref{eq:muPlusMinus_def} have at least 3 continuous derivatives in an open neighborhood of $t$.
    \item The conditional variance functions $\sigma_{\pm}^2(x) \equiv \text{Var}\big(Y_{\pm} \giv X=x \big)$ are both bounded in a neighborhood of $t$ and continuous at $t$.
\end{compactenum}
Under these conditions, $\tau_{\thresh}$ can be consistently estimated by local linear regression with $O_p(N^{-2/5})$ errors \citep{ImbensKalyanaraman_optimalBW}. On the other hand, the conditions above do not suffice to let an RDD consistently estimate $\tau(x)$ for any $x\ne t$.

If an investigator runs a TBD with $\Delta>0$, then assumptions like those above replacing $t$ by $x$ allow consistent estimation of $\tau(x)$ for any $x\in(t-\Delta,t+\Delta)$.
Furthermore, as long as $\var(Y_{\pm})<\infty$,
$$\tau_{\text{ATE}}(\Delta) \equiv 
\e( \tau(X) \giv t-\Delta <X<t+\Delta)
$$ 
can be consistently estimated with error $O_p(N^{-1/2})$ in a TBD without requiring assumptions (i), (ii), and (iii). 

The discussion above leaves open the possibility that an RDD could be better than a TBD 
when estimating $\tau_{\thresh}$. Therefore, for the remainder of the paper our primary focus will be on showing that even if the only goal is estimating $\tau_{\thresh}$, it is still beneficial to run a TBD rather than an RDD. Our other focus will be to show that when the only goal is to estimate $\tau_{\thresh}$, it is beneficial to pick a larger $\Delta$ in the experimental design stage when the option is available. Picking a larger $\Delta$ has other benefits as well such as making $\tau(x)$ identifiable for more values of the assignment variable and making $\tau_{\text{ATE}}(\Delta)$ more representative of the overall population and easier to estimate.  Naturally, there are
non-statistical reasons to keep $\Delta$ smaller.

\subsection*{Local linear estimation}

In keeping with current RDD practice, we suppose that {under an RDD} $\tau_{\thresh}$ will be estimated with local linear regression. In particular, we assume that a parameter vector $\beta$ defined by
 \begin{align}\label{eq:kernelcriterion}
 \hat{\beta} = \argmin\limits_{\beta \in \real^4 }
 \sum\limits_{i=1}^N K \Big( \frac{x_i - t}{h} \Big)\big( Y_i -  (\beta_1 + \beta_2x_i+\beta_3 Z_i +\beta_4 x_i Z_i)  \big)^2
\end{align} will be fit for some symmetric kernel function $K(\cdot)\ge0$ and bandwidth parameter $h>0$, and that $\tau_{\thresh}$ will be estimated with \begin{equation}\label{eq:def_hat_tauthresh}
\hat{\tau}_{\thresh}=2\hat{\beta}_3+2\hat{\beta}_4 t.\end{equation} 
While this formulation of estimating $\hat{\tau}_{\thresh}$ may be less familiar than the approach of fitting separate local linear regressions for the treatment and control groups, it is easy to check that the two formulations yield the same estimator.

Throughout the paper, we suppose that under a TBD, $\hat{\tau}_{\thresh}$ will also be estimated using local linear regression according to \eqref{eq:kernelcriterion} and \eqref{eq:def_hat_tauthresh}. 
We do not use the same bandwidth $h$ for the TBD and RDD.
Indeed, in the next section we see that the optimal bandwidth choice (in terms of AMSE) is different for the two designs.

Because kernels with unbounded support are not typically used in RDD analysis (\cite{cattaneo2022regression}),  we  only consider kernels with bounded support. We assume without loss of generality, that the kernel is supported on $[-1,1]$. 
We have a special interest in a uniform (boxcar) kernel $K_{\bc}(x)=1_{|x|\le 1}$ because it is a popular kernel choice and is a local version of the regression model~\eqref{eq:ovmodel}. We are also interested in a triangular spike kernel $K_{\ts}(x)=(1-|x|)_+$ where $z_+=\max(0,z)$. This kernel was shown by \cite{cheng1997automatic} to optimize a bias-variance tradeoff for extrapolation from $x_i>t$ to $\e(Y\giv x=t)$ and has been advocated for RDD analysis by \cite{ImbensKalyanaraman_optimalBW} and \cite{calonico2014robust} among others.

The local linear regression estimator from~\eqref{eq:def_hat_tauthresh} has a bias and variance that both depend on the bandwidth $h$. Larger $h$ typically bring greater bias because the true regression is not precisely linear over a region centered on $t$.  Smaller $h$ bring greater variance because then fewer data points are in the regression. \cite{ImbensKalyanaraman_optimalBW} develop a method for choosing the bandwidth $h$ that is asymptotically mean squared optimal for the RDD. In the next section, we compare the AMSE of the TBD with that of the RDD, when each of them has their asymptotically optimal bandwidth choice.

In this paper, we focus on the accuracy of the estimated treatment
effect.  The RDD literature includes several papers devoted
to the construction of confidence intervals.
There it is necessary to account for the bias in a local polynomial
regression.  A simple approach is to choose $h$ to undersmooth the
regression function, resulting
in a bias of lower order than the standard error, and this simplifies
confidence interval construction.  Undersmoothing, however, brings less
accuracy \citep{calonico2014robust}.  See \cite{calo:catt:farr:2019}
for a discussion of bandwidth choices to optimize estimation, or
optimize confidence interval construction, or to get robust 
(asymptotically valid) confidence intervals using the bandwidth
that is optimal for estimation.

\section{Asymptotic mean square optimal error}\label{sec:bandwidth_shrink_asymptotics}

In this section, we demonstrate the advantage of the TBD over the RDD 
when each design's bandwidth is chosen to 
minimize the AMSE in the estimation of $\tau_{\thresh}=\mu_+(t)-\mu_-(t)$. Following \cite{ImbensKalyanaraman_optimalBW}, we assume the following more general regularity conditions for estimating the causal effect at $X=t$: \begin{compactenum}[(i)]
    \item The triples $\big(X_i,Y_{i+}, Y_{i-}\big)$ for $i=1,\dots,N$ are IID.
    \item The distribution of $X_i$ has density $f(\cdot)$, which is continuously differentiable at $t$ with $f(t)>0$.
    \item Conditional means $\mu_{\pm}(\cdot)$ both have at least three continuous derivatives in an open neighborhood of $t$, with the $k$'th derivatives at $t$ denoted $\mu_\pm^{(k)}(t)$. 
    \item The kernel $K(\cdot)$ is nonnegative, symmetric, bounded, has support $[-1,1]$, is continuous on its support, and is strictly positive somewhere.
    \item The conditional variances $\sigma_\pm^2(x) \equiv \var (Y_{i \pm} \giv X_i=x )$
    are both bounded in an open neighborhood of $t$ and are continuous and strictly positive at $t$.
    \item $\mu_+^{(2)}(t) \neq \mu_-^{(2)}(t)$.
\end{compactenum}

Under an RDD,  $Z_i=1$ if $X_i> t$
and is $-1$ otherwise,
so these assumptions imply Assumptions 3.1--3.6 that \cite{ImbensKalyanaraman_optimalBW} make for an RDD. To allow for analysis in the TBD setting, our assumptions 
(i)--(vi) are slightly 
stronger than those in \cite{ImbensKalyanaraman_optimalBW}. 
For example, unlike in our Assumption (iii), \cite{ImbensKalyanaraman_optimalBW} make no assumptions on $\mu_+(\cdot)$ in the interval $(-\infty,t)$ or on $\mu_-(\cdot)$ in the interval $(t, \infty)$.
 Regarding assumption (vi), \cite{ImbensKalyanaraman_optimalBW} also consider the case where $\mu_+^{(2)}(t) = \mu_-^{(2)}(t)$ and show that in this case, their proposed method of estimating $\tau_{\thresh}$ has error $O_p(N^{-3/7})$ rather than $O_p(N^{-2/5})$. We do not consider the case where $\mu_+^{(2)}(t) = \mu_-^{(2)}(t)$ in detail for the TBD as the result should be similar to that for the RDD and is of less interest for our head-to-head comparison of TBD with RDD.

Because our assumptions (i)--(vi) imply Assumptions 3.1--3.6 in \cite{ImbensKalyanaraman_optimalBW} for an RDD, if we let  
$$\tilde{\nu}_j \equiv \int_0^{\infty} u^j K(u) \rd u  \quad \text{ and } \quad \tilde{\pi}_j \equiv \int_0^{\infty} u^j K^2(u) \rd u$$ 
for $j \in \mathbb{N}$, and let \begin{align}\label{eq:defctilde}
\tilde{C}_1 \equiv \frac{1}{4} \Big( \frac{\tilde{\nu}_2^2 - \tilde{\nu}_1 \tilde{\nu}_3}{\tilde{\nu}_0 \tilde{\nu}_2 - \tilde{\nu}_1^2} \Big)^2  \quad \text{ and } \quad \tilde{C}_2 \equiv \frac{\tilde{\nu}_2^2 \tilde{\pi}_0 -2 \tilde{\nu}_1 \tilde{\nu}_2 \tilde{\pi}_1+\tilde{\nu}_1^2 \tilde{\pi}_2}{( \tilde{\nu}_2 \tilde{\nu}_0 - \tilde{\nu}_1^2)^2},
\end{align}
and define
\begin{equation}\label{eq:AMSE_RDD_def}
    \text{AMSE}_{\text{RDD}}(h,N) \equiv \tilde{C}_1 \big( \mu_+^{(2)}(t) -\mu_-^{(2)}(t) \big)^2 h^4+ \frac{\tilde{C}_2}{Nh} \Big( \frac{\sigma_+^2(t)+\sigma_-^2(t)}{f(t)} \Big),
\end{equation}  
then Lemma 3.1 of \cite{ImbensKalyanaraman_optimalBW} holds. We reproduce the statement of this lemma below.

 \begin{lemma}
 \label{lemma:MSE_RDD_IK} Under Assumptions (i)--(vi), if an \emph{RDD} determines the treatment assignment and both $h \to 0$ and $Nh \to \infty$ as the number of samples $N \to \infty$, then the mean squared error in estimating $\tau_{\emph{thresh}}$ is given by \begin{equation}\label{eq:asymp_MSE_RDD}
     \emph{MSE}_{\emph{RDD}}(h,N)= \emph{AMSE}_{\emph{RDD}}(h,N) + o_p\Big( h^4 + \frac{1}{Nh} \Big),
 \end{equation} 
 and the asymptotically optimal bandwidth, defined by $\argmin_h \emph{AMSE}_{\emph{RDD}}(h,N)$ is given by \begin{equation}\label{eq:hopt_RDD}  h_{\emph{opt,RDD}}(N)=  \Big( \frac{\tilde{C}_2}{4 \tilde{C}_1} \Big)^{1/5}  \Big( \frac{\sigma_+^2(t) +\sigma_-^2(t)}{f(t) \big( \mu_+^{(2)}(t)- \mu_-^{(2)}(t)\big)^2} \Big)^{1/5}  N^{-1/5}.  
 \end{equation}
 \end{lemma}
\begin{proof}
\citet[Lemma 3.1]{ImbensKalyanaraman_optimalBW}.
\end{proof}

Because we wish to compare the RDD to the tie-breaker design, we derive a similar result for the asymptotic MSE for the tie-breaker design. The TBD counterparts to the RDD quantities above are
\begin{equation}\label{eq:def_nuAndPi_newer}
    \nu_j \equiv \int_{-\infty}^{\infty} u^j K(u) \rd u  \quad \text{ and } \quad \pi_j \equiv \int_{-\infty}^{\infty} u^j K^2(u) \rd u
\end{equation}
for $j \in \mathbb{N}$, 
\begin{equation}\label{eq:C_def}C_1 \equiv \frac{1}{4} \Big( \frac{\nu_2^2 - \nu_1 \nu_3}{\nu_0 \nu_2 - \nu_1^2} \Big)^2  \quad \text{ and } \quad C_2 \equiv  \frac{2\nu_2^2 \pi_0 -4 \nu_1 \nu_2 \pi_1+2\nu_1^2 \pi_2}{( \nu_2 \nu_0 - \nu_1^2)^2},
\end{equation} 
and  
\begin{equation}\label{eq:AMSE_TBD_def}
    \text{AMSE}_{\text{TBD}}(h,N) \equiv C_1 \big(  \mu_+^{(2)}(t) -\mu_-^{(2)}(t) \big)^2 h^4+ \frac{C_2}{Nh} \Big( \frac{  \sigma_+^2(t)+\sigma_-^2(t)}{f(t)} \Big).
\end{equation}

 \begin{lemma}\label{lemma:MSE_TBD} Under Assumptions (i)--(vi), if a \emph{TBD} with a fixed experimental radius $\Delta >0$ determines the treatment assignment and both $h \to 0$ and $Nh \to \infty$ as the number of samples $N \to \infty$,  then the mean squared error in estimating $\tau_{\emph{thresh}}$ is given by \begin{equation}\label{eq:asymp_MSE_TBD}
     \emph{MSE}_{\emph{TBD}}(h,N)= \emph{AMSE}_{\emph{TBD}}(h,N) + o_p\Big( h^4 + \frac{1}{Nh} \Big),
 \end{equation} and the asymptotically optimal bandwidth, defined by $\argmin_h \emph{AMSE}_{\emph{TBD}}(h,N)$ is \begin{equation}\label{eq:hopt_TBD}  h_{\emph{opt,TBD}}(N)=  \Big( \frac{C_2}{4 C_1} \Big)^{1/5}  \Big( \frac{  \sigma_+^2(t) +\sigma_-^2(t)}{f(t) \big( \mu_+^{(2)}(t)- \mu_-^{(2)}(t)\big)^2} \Big)^{1/5}  N^{-1/5}.  \end{equation}
 \end{lemma}
 
 \begin{proof}
 See Appendix \ref{sec:proof_of_TBD_MSE}.
 \end{proof}
 
 The proof of this lemma is very similar to the proof of Lemma 3.1 in \cite{ImbensKalyanaraman_optimalBW}, from their appendix. Instead of pointing to their proof and noting the parts of their proof that differ in the tie-breaker design setting, we write out the proof of Lemma~\ref{lemma:MSE_TBD} in Appendix~\ref{sec:proof_of_TBD_MSE} to ensure there are no subtle issues with using their proof in the tie-breaker design setting.

{The leading order MSE formulas are derived by evaluating and summing the leading order terms for both the squared-bias and the variance.} In formulas \eqref{eq:AMSE_RDD_def} and \eqref{eq:AMSE_TBD_def} for the leading order MSE, the first term gives the leading order squared-bias while the second term gives the leading order variance. See formulas \eqref{eq:as_bias_TBD} and \eqref{eq:TBD_as_var_formula} for explicit calculations of the leading order bias and variance in the TBD case, and see the formulas for `B' and `V' in the appendix of \cite{ImbensKalyanaraman_optimalBW} for explicit calculations of these quantities in the RDD case. It is not surprising that the formulas for the leading order squared-bias, variance and MSE are different for the two design types because for an RDD, estimation of $\tau_{\thresh}$ involves estimation of the mean functions at a boundary point, whereas for a TBD, estimation of $\tau_{\thresh}$ involves estimation of mean functions at an interior point. 
 
 In Figure \ref{fig:bias_var_tradeoff}, we plug in scalar multiples of the optimal bandwidth for the RDD given in Lemma \ref{lemma:MSE_RDD_IK} to the first and second terms of formulas \eqref{eq:AMSE_RDD_def} and \eqref{eq:AMSE_TBD_def} to visualize the trade-off for the leading order squared-bias and variance in a tie-breaker design compared to a regression discontinuity design. The formulas simplify when defining the quantity \begin{equation}\label{eq:alpha_definition}
     \alpha \equiv \frac{5}{4} \big(  \mu_+^{(2)}(t) - \mu_-^{(2)}(t) \big)^{2/5} \Big( \frac{  \sigma_+^2(t) + \sigma_-^2(t) }{  f(t)}     \Big)^{4/5}, 
 \end{equation} which does not depend on $h$, $N$ or the kernel choice.

\begin{figure}[t]
\centering
\includegraphics[width=1 \hsize]{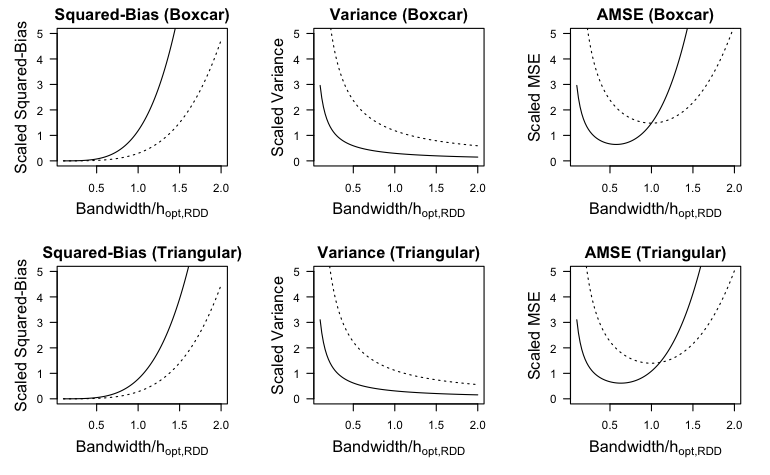}
\caption{\label{fig:bias_var_tradeoff} A comparison of the asymptotic bias-variance tradeoff for the regression discontinuity design (dotted lines) versus for the tie-breaker design (solid lines). The x-axes are in units of asymptotically MSE optimal bandwidth for RDD given at \eqref{eq:hopt_RDD} while the y-axes are the leading order terms in units of $\alpha N^{-4/5}$ where $N$ is the sample size and $\alpha$ is a constant given in \eqref{eq:alpha_definition} that depends on properties of the joint distribution of $(X,Y,Z)$ in a neighborhood of the cutoff. }
\end{figure}

In practice, the optimal bandwidth is not known and must be estimated. For both the RDD and the TBD, the optimal bandwidth depends on the quantity \begin{equation}\label{eq:def_gamma}
    \gamma \equiv \Big( \frac{\sigma_+^2(t) +\sigma_-^2(t)}{f(t) \big( \mu_+^{(2)}(t)- \mu_-^{(2)}(t)\big)^2} \Big)^{1/5}
\end{equation} which must be estimated from the observed data. We consider the regularized estimator for $\gamma$ of \cite{ImbensKalyanaraman_optimalBW}. We take the estimated optimal bandwidth $\hat{h}_{\text{opt}}$ proposed in their Section 4.2 and set $\hat{\gamma}_{\text{RDD}}= (4 \tilde{C}_1/\tilde{C_2} )^{1/5}\hat{h}_{\text{opt}} N^{1/5}$. It can be seen from the proof of Theorem 4.1 in \cite{ImbensKalyanaraman_optimalBW} that under assumptions (i)--(vi), $\hat{\gamma}_{\text{RDD}} \xrightarrow{p} \gamma$.
In the TBD case, we know a consistent estimator of $\gamma$ exists. For example, if we let $\hat{\gamma}_{\text{TBD,naive}}$ be an estimator of $\gamma$ that is constructed similarly to $\hat{\gamma}_{\text{RDD}}$ using only the subset of the data which looks like an RDD, $\hat{\gamma}_{\text{TBD,naive}} \xrightarrow{p} \gamma$. Of course such an estimator of $\gamma$ is inefficient; in practice one should instead use an estimator of $\gamma$ that does not throw out all of the control samples for which $x>t$ and all of the treated samples for which $x<t$. For our theoretical comparison of TBDs with RDDs, we are not concerned with the actual form of $\hat{\gamma}_{\text{TBD}}$ as long as it is consistent. Therefore, in the TBD case we will let $\hat{\gamma}_{\text{TBD}}$ be any estimator that satisfies $\hat{\gamma}_{\text{TBD}} \xrightarrow{p} \gamma$. We make a few remarks about estimation of $\gamma$ in the TBD setting in the discussion section.

To compare the AMSE for the RDD versus the TBD, we will assume that if the investigator were to run an RDD and were seeking mean squared optimal estimation of $\tau_{\thresh}$, they would ultimately use the bandwidth \begin{equation}\label{eq:hopt_RDD_est}
    \hat{h}_{\text{opt,RDD}}(N) = \Big( \frac{\tilde{C}_2}{4 \tilde{C}_1} \Big)^{1/5}  \hat{\gamma}_{\text{RDD}}  N^{-1/5},
\end{equation} 
where $\hat{\gamma}_{\text{RDD}}$ is the consistent estimator for $\gamma$ described above and $\tilde{C}_j$ are defined at~\eqref{eq:defctilde}. 
We will also assume that if the investigator were to run a TBD seeking mean squared optimal estimation of $\tau_{\thresh}$, they would ultimately use the bandwidth \begin{equation}\label{eq:hopt_TBD_est}
    \hat{h}_{\text{opt,TBD}}(N) = \Big( \frac{C_2}{4 C_1} \Big)^{1/5}  \hat{\gamma}_{\text{TBD}}  N^{-1/5},
\end{equation} where $\hat{\gamma}_{\text{TBD}}$ is any consistent estimator of $\gamma$ and $C_j$ are defined at \eqref{eq:C_def}.

The following theorem compares the RDD with $N$ points to a TBD 
with $\theta N$ points for some $\theta>0$. We will use the value of $\theta$ that provides equal MSEs for estimation of $\tau_{\thresh}$ as a metric to compare the two designs. 

\begin{theorem}\label{theorem:MSE_RDD_versu_TBD}
Let $\theta>0$ be a constant. Under assumptions (i)--(vi), as $N \to \infty$
\begin{equation}\label{eq:MSE_rat}
    \frac{\emph{MSE}_{\emph{RDD}}\big(\hat{h}_{\emph{opt,RDD}}(N),N \big)}{\emph{MSE}_{\emph{TBD}}\big(\hat{h}_{\emph{opt,TBD}}(\theta N), \theta N \big)} \xrightarrow{p} \theta^{4/5} \Big( \frac{\tilde{C}_1 \tilde{C}_2^4}{C_1 C_2^4}\Big)^{1/5} 
\end{equation}
holds for any tie-breaker design
of the form \eqref{eq:threelevel} with $\Delta>0$. 
\end{theorem}

\begin{proof}
See Appendix~\ref{sec:relMSETheoremProof}.
\end{proof}

Theorem~\ref{theorem:MSE_RDD_versu_TBD} uses the 
assumption that $\Pr(Z_i=1 \giv x_i )=1/2$ for $x_i$ in
the randomization window.
If $\sigma^2_+(t)\ne\sigma^2_-(t)$, then we might
prefer to offer the treatment with probability
$p\ne1/2$.
In Appendix \ref{sec:GenPneqHalf}, we study 
a treatment probability $p\in(0,1)$.
When $p = \sigma_+(t)/(\sigma_+(t) +\sigma_-(t))$, the asymptotic MSE is minimized, though
an investigator would also want to account for the
cost of the treatment.  If one chooses $p$
using poor prior estimates of $\sigma_\pm$ it is
possible that the resulting TBD will have
a higher asymptotic MSE than the RDD.  However, for any of the kernels
in Table~\ref{table:Kernels}, one can
protect against that by choosing
$p\in[0.18,0.82]$.

\subsection*{Asymptotic MSE comparison for some specific kernels}

We now use Theorem~\ref{theorem:MSE_RDD_versu_TBD} to compare the MSE in estimating $\tau_{\thresh}$ for the RDD versus the TBD, under optimal bandwidth choices for various kernels of interest. See Table \ref{table:Kernels}. If an investigator is deliberating between an RDD with $N$ samples versus conducting a TBD (for a fixed $\Delta>0$) with $N$ samples, and either experimental design is to be analyzed with the asymptotically optimal bandwidth choice for the prespecified kernel, 
then the ratio of the MSEs will converge in probability to $\big((\tilde{C}_1 \tilde{C}_2^4)/(C_1 C_2^4)\big)^{1/5}$ as $N \to \infty$. Using formulas \eqref{eq:defctilde} and \eqref{eq:C_def}, the fourth column of Table \ref{table:Kernels} gives the value of the quantity $\big((\tilde{C}_1 \tilde{C}_2^4)/(C_1 C_2^4)\big)^{1/5}$ rounded to 2 decimal places. For the boxcar and triangular kernels respectively, this quantity is precisely $64^{1/5}$ and $60.46618^{1/5}$ without rounding. 

It is also interesting to consider the quantity given by \begin{equation}\label{eq:DefthetaStar}
    \theta_* = \frac{C_1^{1/4}C_2}{\tilde{C}_1^{1/4} \tilde{C}_2}.
\end{equation}
As a result of Theorem \ref{theorem:MSE_RDD_versu_TBD}, an experimental designer deciding to use a TBD rather than an RDD would only need to collect $\theta_*$ times as many samples in order to achieve the same asymptotic MSE in estimating $\tau_{\thresh}$.

\begin{table}[!t]
\caption{An asymptotic comparison of the regression discontinuity designs with tie-breaker designs in Kernel regression-based estimation of $\tau_{\thresh}$. The fourth column gives the quantity $\big((\tilde{C}_1 \tilde{C}_2^4)/(C_1 C_2^4)\big)^{1/5}$, which is computed using \eqref{eq:defctilde} and \eqref{eq:C_def} and rounded to 2 decimal places. The last column gives the quantity $\theta_*$ given by \eqref{eq:DefthetaStar} rounded to 3 decimal places.}
\label{table:Kernels}
\begin{center}
\begin{tabular}{ l l c c c } 
\toprule
 Kernel & Function $K(u)$ & Support & Relative AMSE & $\theta_*$  \\ 
\midrule
 Boxcar & $1/2$ & $[-1,1]$  & 2.30 & 0.354 \\ [1ex] 

Triangular & $(1-\vert u \vert)_+$ & $[-1,1]$  & 2.27 &  0.359 \\  [1ex] 

 Epanechnikov & $ \frac{3}{4}(1- u^2)_+$ & $[-1,1]$ & 2.31 & 0.351 \\ [1ex] 

 Quartic & $ \frac{15}{16} (1-u^2)_+^2$  & $[-1,1]$ & 2.29 & 0.355  \\ [1ex] 

 Triweight & $ \frac{35}{32} (1-u^2)_+^3$  & $[-1,1]$ & 2.28 & 0.357  \\ [1ex] 

  
  Tricube & $ \frac{70}{81} (1-\vert u\vert^3 )_+^3$ & $[-1,1]$ & 2.31 & 0.352  \\ [1ex] 

 
 Cosine & $\frac{\pi}{4} \cos \big( \frac{\pi}{2} u \big)$  & $[-1,1]$ & 2.31 & 0.352 \\ 
 \bottomrule
\end{tabular}
\end{center}
\end{table}

Table \ref{table:Kernels} shows that the kernel choice has a remarkably small impact on the relative benefit of using a TBD rather than an RDD to estimate $\tau_{\thresh}$. It is well known in the usual kernel smoothing setting that there is little difference in performance among the widely used kernels. See \cite{wand:jone:1994}.
 If $\tau_{\thresh}$ is to be estimated with local linear regression using one of the seven popular kernel choices exhibited in Table \ref{table:Kernels}, then the RDD has an asymptotic MSE that is about 2.3 times as large as that of the TBD, and the TBD will require 64 to 65 percent fewer samples than the RDD in order to achieve the same asymptotic MSE. 

We think that a version of Theorem~\ref{theorem:MSE_RDD_versu_TBD} 
will hold also for unbounded kernels such as the $\dnorm(0,1)$ density under
reasonable but stronger regularity conditions on $f(\cdot)$, $\mu_\pm(\cdot)$
and $\sigma_\pm(\cdot)$. We do not develop such a result
as \cite{cattaneo2022regression} state that kernels with unbounded
support are not used in RDD analysis.

\section{
Variance comparisons at fixed $h> \Delta$} \label{sec:AsymApprox_via_integrals}

The AMSE comparison in Section \ref{sec:bandwidth_shrink_asymptotics} depends upon the optimal TBD bandwidth, $\hat{h}_{\text{opt,TBD}}$, eventually becoming smaller than the
positive experimental radius $\Delta$. 
However, the optimal $h$ converges to zero only at the very
slow rate $N^{-1/5}$.  Furthermore, the constant in
that rate includes the factor
$|\mu_+^{(2)}(t)-\mu_-^{(2)}(t)|^{-2/5}$
which could be very large.
We believe that in many applied settings the optimal
value of $h$ will not be smaller than $\Delta$.
Then $\Delta/h$ is not necessarily
within the support of the kernel and $Y$ values
data from outside the experimental region are included in the local linear regression.

In this section, we complement the
prior analysis with one where $h$ is fixed
and larger than $\Delta$.
We assume a symmetric
kernel function that is Lipschitz continuous on its support.

The kernel regression estimate of
$\hat\tau_{\thresh}$ has a leading bias of $O(h^2)$.
In the regime where the bandwidth is bigger than $\Delta$, mean squared optimality analysis for estimating $\tau_{\thresh}$ is complicated by the fact that for the TBD there will often exist an $h> \Delta$ such that the constant in this $O(h^2)$ term vanishes.
Remarkably, such a bandwidth depends only on the experimental radius $\Delta$ and the kernel $K$. It does not depend on $\mu_+$, $\mu_-$, $f$, or $N$. In Appendix \ref{sec:MagicBandwidth}, we prove that under certain regularity conditions on $f$, $\mu_{\pm}$, and $K$, a bandwidth $h$ that solves $\nu_2^2 = 4 \int_{\Delta/h}^{\infty} uK(u) \rd u  \int_{\Delta/h}^{\infty} u^3K(u) \rd u$ removes the leading order bias, and moreover, such a solution exists. See Table \ref{table:MagicBandwidth} for numerical solutions of this equation for the kernel choices considered previously. We find that for these kernel choices, the bandwidth removing the leading order bias ranges from approximately $3.13 \Delta$ for the Boxcar kernel to approximately $4.84 \Delta$ for the Triweight kernel. We caution investigators against picking this bandwidth 
because it does not shrink with $N$. It
could place too little weight on reducing variance for small $N$ and the third order bias term will be $O(1)$.

Due to the existence of a fixed bandwidth bigger than $\Delta$ that removes the leading order bias of $\hat{\tau}_{\thresh}$, analysis of bias and mean squared optimality using second order Taylor expansions of $\mu_{\pm}(\cdot)$ would be misleading. Hence, we do not conduct an analysis similar to that seen in Section \ref{sec:bandwidth_shrink_asymptotics} for the regime where $h>\Delta$. 
For that regime, we instead restrict our attention to the variance in estimating $\tau_{\thresh}$ at a fixed bandwidth $h$. 

The variance of the local linear estimator $\hat{\tau}_{\thresh}$ given in \eqref{eq:kernelcriterion} and \eqref{eq:def_hat_tauthresh} can be computed as follows. The design matrix for the regression is $\cx\in\real^{N\times 4}$ with $i$'th
row $(1,x_i,Z_i,x_iZ_i)$.  The response is $\cy = (Y_1,\dots,Y_N)^\tran$.
For simplicity we assume $\var(\cy \giv \cx)=\sigma^2 I_{N}$ and without loss of generality we assume $t=0$. The kernel weights are $K(x_i/h)$, and we
let $\cw=\cw(h)\in\real^{N\times N}=\diag(K(x_i/h))$. Then
\begin{align}\label{eq:betahat}
\hat\beta=\hat\beta(\Delta)
=(\cx^\tran\cw\cx)^{-1}\cx^\tran\cw\cy,
\end{align}
and under the assumption that $\var(\cy \giv \cx)=\sigma^2 I_{N}$ we have
\begin{align}\label{eq:varbetahat}
\var(\hat\beta\giv \cx;\Delta) = (\cx^\tran\cw\cx)^{-1}\cx^\tran\cw^2\cx(\cx^\tran\cw\cx)^{-1}\sigma^2.
\end{align}
Formula~\eqref{eq:betahat}  for $\hat\beta$ matches the familiar generalized least
squares formula for the case where $\var(\cy\giv \cx)=\cw\sigma^2$.  Here $\cw$
arises from weights that are not of inverse variance type and hence the formula for $\var(\hat\beta\giv\cx;\Delta)$ involves a $\cw^2$ factor and less cancellation than we might have expected. The boxcar kernel is special because then $K(x_i/h)\in\{0,1\}$ equals its own square. In that case $\var(\hat\beta\giv\cx;\Delta)=(\cx^\tran\cw\cx)^{-1}\sigma^2$. The estimator is $\hat{\tau}_{\thresh} = 2\hat{\beta}_3$. Therefore, we study $\var(\hat\beta_3\giv\cx;\Delta)$ under a tie-breaker design as $(\var(\hat\beta\giv\cx;\Delta))_{3,3}$ using the expression in~\eqref{eq:varbetahat}.

At the stage where the experiment is being designed and $\Delta$ is being chosen, the investigator does not have much information about $\cx\in\real^{N\times 4}$ but we will later see, quite a bit is known about the quantity $N \times \var(\hat\beta_3\giv\cx;\Delta)/\sigma^2$. For $x_i$ from a real dataset, we see in Section \ref{sec:Data_application} (e.g. Figure \ref{fig:MC_boxcar_plots_AL}) that this quantity does not vary much for different simulations of the random treatment assignments $(Z_i)_{i=1}^N$. To get theoretical insight, we turn our attention to the uniformly spaced setting with $x_i=(2i-N-1)/N$ to develop tractable theoretical results. We give an asymptotic justification
for this assumption using results from
\cite{fan1996local} in Section~\ref{sec:Data_application}. This rank transformation is also used in \cite{owen:vari:2020}.

For $x_i = \xiVal$, the matrices $\cx^\tran\cw\cx/N$ and $\cx^\tran\cw^2\cx/N$ contain elements that can be approximated by integrals of the form
\begin{align}
\ci_{}^{rst}=\ci^{rst}(\Delta,h,K)\equiv\frac12\int_{-1}^1x^r\e(Z^s\giv x;\Delta)K\Bigl(\frac{x}h\Bigr)^t\rd x
\end{align}
for integer exponents $r$, $s$ and $t$.
Our expressions will simplify somewhat because $Z^2=1$ making every $\cixzk^{r,2,t}=\cixzk^{r,0,t}$ and also because both $x$ and $\e(Z\giv x;\Delta)$ are antisymmetric functions of $x$ making them orthogonal to $K(x/h)$ which we have assumed to be symmetric.
The error in those moment approximations is
$O_p(N^{-1/2})$ if the $Z_i$ are independent
random variables. The error can be much less with
other sampling schemes. For instance, we could use stratified sampling, forming pairs of subjects $(i,i+1)$ in the experimental region and randomly setting $Z_i=\pm1$ and $Z_{i+1}=-Z_i$.
We will use $\approx$ to describe approximations that are $O_p(N^{-1/2})$ or better. 

Applying first $Z^2=1$ and then using symmetry and
anti-symmetry
\begin{align*}
\frac1N\cx^\tran\cw\cx &\approx
\kbordermatrix{ & 1 & x & z & xz \\
1 &\cixzk^{001} & \cixzk^{101} & \cixzk^{011} &\cixzk^{111}\\[1ex]
x &\cixzk^{101} & \cixzk^{201} & \cixzk^{111}&\cixzk^{211}\\[1ex]
z &\cixzk^{011} & \cixzk^{111} & \cixzk^{021} & \cixzk^{121}\\[1ex]
xz & \cixzk^{111} &\cixzk^{211} & \cixzk^{121} & \cixzk^{221}
 }
&=
\begin{bmatrix}
\cixzk^{001} & 0 & 0 &\cixzk^{111}\\[1ex]
0 & \cixzk^{201} & \cixzk^{111}&0\\[1ex]
0& \cixzk^{111} & \cixzk^{001} & 0\\[1ex]
\cixzk^{111} & 0 &0 & \cixzk^{201}
\end{bmatrix}.
\end{align*}
Because $K^2(\cdot)$ is also a symmetric function we also get
$$
\frac1N\cx^\tran\cw^2\cx \approx
\begin{bmatrix}
\cixzk^{002} & 0 & 0 &\cixzk^{112}\\[1ex]
0 & \cixzk^{202} & \cixzk^{112}&0\\[1ex]
0& \cixzk^{112} & \cixzk^{002} & 0\\[1ex]
\cixzk^{112} & 0 &0 & \cixzk^{202}
\end{bmatrix}.
$$

From all of the symmetries involved in the 32 components of these two matrices, we need to consider at most six distinct integrals.
We rewrite those matrices, beginning with
 \begin{align}\label{eq:xtwxbyn} \frac{1}{N}  \cx^T \cw  \cx \approx \begin{bmatrix} \kappa_0 & 0 & 0 & \phi(\Delta) \\ 0 & \kappa_2 & \phi(\Delta) & 0 \\ 0 & \phi (\Delta) & \kappa_0 & 0 \\ \phi(\Delta) & 0 & 0 & \kappa_2 \end{bmatrix}
\end{align}
where
\begin{align}\label{eq:defnu}
\begin{split}\kappa_0&
=\frac12\int_{-1}^{1}K\Bigl(\frac{x}h\Bigr)\rd x,\quad
\kappa_2
=\frac12\int_{-1}^{1}x^2K\Bigl(\frac{x}h\Bigr)\rd x,\quad\text{and}\\
\phi(\Delta)&
=\frac12\int_{-1}^{-\Delta}(-x)K\Bigl(\frac{x}h\Bigr)\rd x +\frac12\int_{\Delta}^{1}xK\Bigl(\frac{x}h\Bigr)\rd x=\int_{\Delta}^1xK\Bigl(\frac{x}h\Bigr)\rd x.
\end{split}
\end{align}
Note that $\kappa_0$ and $\kappa_2$ may depend on $h$ but they do not depend on $\Delta$.
A similar argument shows that
\begin{align}\label{eq:xtwwxbyn}   \frac{1}{N}  \cx^T \cw^2  \cx \approx \begin{bmatrix} \lambda_0 & 0 & 0 & \psi(\Delta) \\ 0 & \lambda_2 & \psi(\Delta) & 0 \\ 0 & \psi (\Delta) & \lambda_0 & 0 \\ \psi(\Delta) & 0 & 0 & \lambda_2 \end{bmatrix} \end{align}
for
\begin{align}\label{eq:defpi}
\begin{split}\lambda_0
&=\frac12\int_{-1}^{1}K^2\Bigl(\frac{x}h\Bigr)\rd x,\quad
\lambda_2=\frac12\int_{-1}^{1}x^2K^2\Bigl(\frac{x}h\Bigr)\rd x,\quad\text{and}\\
\psi(\Delta)&=\int_{\Delta}^1xK^2\Bigl(\frac{x}h\Bigr)\rd x.\end{split}
\end{align}
Now we are ready to describe the asymptotic
variance of $\hat\beta_3$.
\begin{theorem}\label{thm:varbetahat3}
Let $x_i= \xiVal$, select $Z_i\in\{-1,1\}$ by the tie-breaker
equation~\eqref{eq:threelevel} with $t=0$. Let $Y_i$ be uncorrelated random
variables with common variance $\sigma^2$, conditionally on
$\cx=( (1,x_1,Z_1,x_1Z_1),\cdots,(1,x_N,Z_N,x_NZ_N))$. Next, for a symmetric kernel
$K(\cdot)\ge0$ that is Lipschitz continuous on its support
and a bandwidth
$h>0$, let $\hat\beta$ be estimated by the kernel weighted
regression~\eqref{eq:kernelcriterion}. Then
\begin{align}\label{eq:varbetahat3}
N\var(\hat\beta_3\giv\cx;\Delta) =
\frac{\sigma^2 \big(\kappa_2^2 \lambda_0-2 \kappa_2 \phi(\Delta) \psi(\Delta) +\lambda_2 \phi^2(\Delta) \big)}{ \big(\kappa_0 \kappa_2 - \phi^2(\Delta)  \big)^{2}}
+O_p\Bigl(\frac1{\sqrt{N}}\Bigr),
\end{align}
where $\kappa_0$, $\kappa_2$ and $\phi(\Delta)$ are defined in~\eqref{eq:defnu} and $\lambda_0$, $\lambda_2$ and $\psi(\Delta)$ are defined in~\eqref{eq:defpi}.
\end{theorem}
\begin{proof}
Reordering the components of $\beta$ we find after substituting equations~\eqref{eq:xtwxbyn} and~\eqref{eq:xtwwxbyn} into~\eqref{eq:varbetahat} that $\sqrt{N}(\hat\beta_1,\hat\beta_4,\hat\beta_2,\hat\beta_3)$ has variance
\begin{align*}
\begin{pmatrix}
\kappa_0 &\phi & 0 & 0\\
\phi & \kappa_2 & 0 & 0\\
0 & 0 & \kappa_0 & \phi\\
0 & 0 & \phi & \kappa_2
\end{pmatrix}^{\!\!-1}
\!\!\!  
\begin{pmatrix}
\lambda_0 &\psi & 0 & 0\\
\psi & \lambda_2 & 0 & 0\\
0 & 0 & \lambda_0 & \psi\\
0 & 0 & \psi & \lambda_2
\end{pmatrix}
\!\!
\begin{pmatrix}
\kappa_0 &\phi & 0 & 0\\
\phi & \kappa_2 & 0 & 0\\
0 & 0 & \kappa_0 & \phi\\
0 & 0 & \phi & \kappa_2
\end{pmatrix}^{\!\!-1}\!\!\sigma^2+O_p\Bigl(\frac1{\sqrt{N}}\Bigr).
\end{align*}
Now~\eqref{eq:varbetahat3} follows directly by matrix inversion and multiplication.
\end{proof}

Our Lipschitz condition on the kernel $K$ is present for technical reasons. Using a formulation with fixed and discrete $x_i=\xiVal$, this condition allows us to obtain the same error rate of $O_p(N^{-1/2})$ as would be obtained using the formulation with random $x_i \stackrel{\text{IID}}{\sim} \dunif[-1,1]$. Without the Lipschitz condition, an adversarially chosen kernel $K(\cdot)$ might have point discontinuities at every rational multiple of the bandwidth $h$. We remark that our Lipschitz condition can be loosened to a $1/2$-Hölder continuity condition, with details available from the first author upon request.

The variance formula in Theorem~\ref{thm:varbetahat3}
does not require the linear model~\eqref{eq:ovmodel}
to hold.
When it does not hold there will generally be some bias
where $\e(2\hat\beta_3\giv\cx;\Delta)\ne \mu_+(0)-\mu_-(0)$.
We suppose that the user will choose an $h$ to appropriately navigate the bias variance tradeoff, but that step takes place after the outcomes $Y_i$ are observed, which are not available when $\Delta$ is chosen, so we resort to comparing the variance for any choice of $h$.

We are primarily interested in comparing the asymptotic variance of $\hat{\tau}_0=2\hat\beta_3$ for various  choices of $\Delta$. We especially want to compare the efficiency of tie-breaker designs with $\Delta>0$ to the RDD with $\Delta=0$. To do this we consider the efficiency ratio
\begin{align}\label{eq:defasyeffN}
\eff^{(N)}(\Delta) \equiv  \frac{\var(\hat{\tau}_{\thresh}\giv\cx; 0)}{\var(\hat{\tau}_{\thresh} \giv \cx; \Delta)}=
\frac{\var(\hat{\beta}_3\giv\cx; 0)}{\var(\hat{\beta}_3 \giv \cx; \Delta)}.
\end{align}

Using Theorem~\ref{thm:varbetahat3}, $\eff^{(N)}(\Delta)$ converges in
probability to the asymptotic efficiency ratio
\begin{align}\label{eq:general_efficiency_ratio} \eff(\Delta) =  \frac{\big(\kappa_2^2 \lambda_0-2 \kappa_2 \phi(0) \psi(0) +\lambda_2 \phi^2(0) \big) \big(\kappa_0 \kappa_2 - \phi^2(\Delta)  \big)^{2} }{ \big(\kappa_2^2 \lambda_0-2 \kappa_2 \phi(\Delta) \psi(\Delta) +\lambda_2 \phi^2(\Delta) \big) \big(\kappa_0 \kappa_2 - \phi^2(0)  \big)^{2}}  \end{align}
using quantities that we defined at~\eqref{eq:defnu} and~\eqref{eq:defpi}.

\subsection*{Efficiency with boxcar and triangular kernels}

In this subsection we present the efficiency ratios under the conditions of Theorem~\ref{thm:varbetahat3} for the two kernels of greatest interest: the boxcar kernel and the triangular kernel. We work with $x_i = \xiVal$ throughout this subsection. 

For the boxcar kernel
$K_{\bc}(u) = 1_{|u|\le 1}$,
we can assume without loss of generality that $h\le 1$ because there are no data with $|x_i-t|=|x_i|>1$, and then any $h>1$ will give the same estimate as $h=1$. We find for this kernel that
\begin{align}\label{eq:boxcarpieces}\kappa_0=\lambda_0= h,\quad \kappa_2=\lambda_2=\frac{h^3}{3},\quad\text{and}\quad \phi(\Delta)=\psi(\Delta)=\frac{(h^2- \Delta^2)_+}{2}.
\end{align}
Using some foresight, we define the local tie-breaker constant $\delta=\Delta/h$. This is the fraction of the local regression region in which the treatment was assigned at random.

\begin{proposition}\label{prop:boxcarefficiency}
Under the conditions of Theorem~\ref{thm:varbetahat3} and using the boxcar kernel $K_{\bc}$, the asymptotic efficiency ratio of the tie-breaker design is
\begin{align}\label{eq:boxcarefficiency}
\eff_{\bc} = 1+6\delta^2-3\delta^4
\end{align}
for $\delta = \Delta/h\le1$. If $\delta>1$, then $\eff_{\bc}=4$.
\end{proposition}
\begin{proof}
Because many quantities from~\eqref{eq:boxcarpieces} are identical, substituting them into \eqref{eq:general_efficiency_ratio} produces numerous simplifications that yield
\begin{align*}
\eff_{\bc} &= \frac{\kappa_0\kappa_2-\phi^2(\Delta)}{\kappa_0\kappa_2-\phi^2(0)}
=\frac{\frac{h^4}3-\frac{(h^2-\Delta^2)_+^2}4}{\frac{h^4}3-\frac{h^4}4}
=4-3(1-\delta^2)_+^2.
\end{align*}
For $0\le\delta<1$ formula~\eqref{eq:boxcarefficiency}  follows from expanding the quadratic while for $\delta>1$ the positive part term vanishes.
\end{proof}

Choosing $h=1$ makes the local regression a global one. We then get the same efficiency ratio as in equation (6) from \cite{owen:vari:2020}. By taking derivatives it is easy to show that the efficiency ratio in~\eqref{eq:boxcarefficiency} is strictly increasing as the local amount of experimentation $\delta$ varies over the interval $0<\delta< 1$.
Figure~\ref{fig:Unif_Efficiecy_plots} plots $\eff_{\bc}$ versus $\delta$.

The triangular spike kernel
$K_{\ts}(x) = (1 -\vert x \vert )_+$
(triangular kernel for short) is more complicated than the boxcar kernel because for it,
$K^2$ is not proportional to $K$.
Once again, we assume that $h \in [0,1]$. For this kernel we compute
$$\kappa_0= \frac{h}{2},\quad \kappa_2 = \frac{h^3}{12},
\quad \lambda_0 = \frac{h}{3},\quad\text{and}\quad \lambda_2=\frac{h^3}{30}$$
and then using $\delta = \Delta/h$, we get
$$\phi(\Delta)= \frac{h^2}{6} (1- 3 \delta^2+ 2 \delta^3)\quad\text{and}\quad\psi(\Delta)= \frac{h^2}{12} (1- 6 \delta^2+ 8 \delta^3 - 3 \delta^4).$$

\begin{proposition}\label{prop:triangleefficiency}
Under the conditions of Theorem~\ref{thm:varbetahat3} and using the triangular kernel $K_{\ts}$, the asymptotic efficiency of the tie-breaker design is
\begin{align}\label{eq:triangleefficiency}
\eff_{\ts} =
\frac{ 2\bigl(3-2(1-3\delta^2+2\delta^3)^2\bigr)^2  }
{ 5-5(1-3\delta^2+2\delta^3)(1-6\delta^2+8\delta^3-3\delta^4)+2(1-3\delta^2+2\delta^3)^2 }
\end{align}
for $\delta = \Delta/h\le1$.
\end{proposition}
\begin{proof}
This follows from plugging in the values of $\kappa_0$, $\kappa_2$, $\lambda_0$, $\lambda_2$, $\phi(\Delta)$, and $\psi(\Delta)$ for the triangular kernel into \eqref{eq:general_efficiency_ratio}. See Appendix \ref{sec:TSefficiencyCalc} for the explicit calculations.
\end{proof}

The second panel in Figure~\ref{fig:Unif_Efficiecy_plots}
shows $\eff_{\ts}$ versus the local experiment size $\delta$. The efficiency curve has a similar monotone increasing shape as we saw for the boxcar kernel.  The maximum efficiency ratio, at $\delta=1$, is $18/5=3.6$ instead of $4$. The efficiency ratio is a rational function of $\delta$ with a numerator of degree $12$ and a denominator of degree $7$.  It is strictly increasing on the interval $0<\delta<1$, though the proof is lengthy enough to move to the Appendix.
\begin{proposition}\label{prop:itsmonotone}
The derivative of $\eff_{\ts}$ with respect to $\delta$ is positive for $0<\delta<1$.
\end{proposition}
\begin{proof}
See Appendix \ref{sec:MonotoneEffTS}.
\end{proof}

\begin{figure}[t!]
    \centering
        \begin{subfigure}[b]{6 cm}
            \centering
            \includegraphics[width=6 cm]{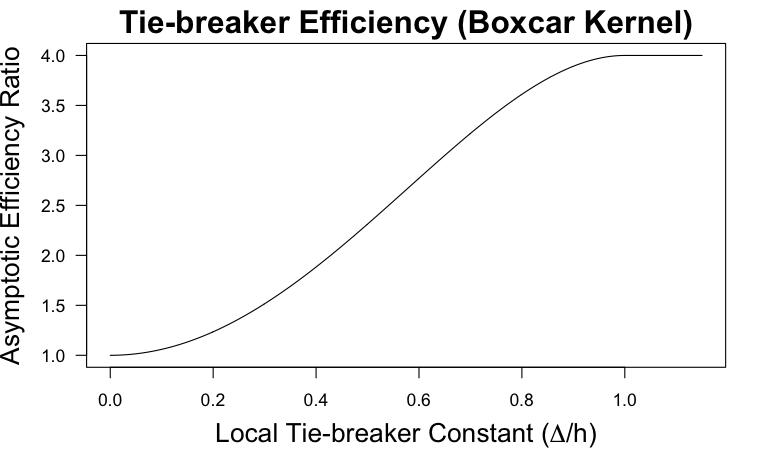}
            \label{fig:boxcarERUnifPlot}
        \end{subfigure}
        \begin{subfigure}[b]{6 cm }
            \centering
            \includegraphics[width=6 cm]{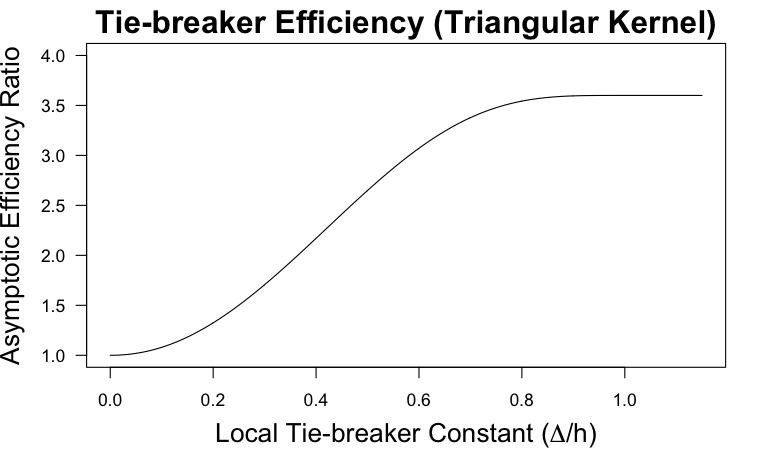}
            \label{fig:triangularERUnifPlot}
        \end{subfigure}
        \caption[]
        {\label{fig:Unif_Efficiecy_plots}
        The left panel shows the efficiency ratio of the tie-breaker design for uniform $x_i$ and the boxcar kernel as a function of $\delta=\Delta/h$.  The right panel shows this efficiency ratio for the triangular kernel.
        }
\end{figure}

\section{Classroom size data} \label{sec:Data_application}

We explored the efficiency ratio for the tie-breaker design for $x_i$ with a uniform distribution.  While that can be arranged by using ranks, in other situations we might prefer to use the original value of a running variable and those might not be uniformly distributed.
We show how to do this using a dataset from \cite{AngristLavy} on classroom sizes.

\cite{AngristLavy} studied the causal effect of classroom size on test performance of elementary school students in Israel. In Israel,
the Maimonides rule mandates that elementary school classes cannot exceed 40 students. If a school has 41 students enrolled in a particular grade that grade must be split into two classes. Note that grades that have 40 or fewer enrolled students are allowed to split into multiple classes and that grades with slightly more than 40 students occasionally violate the Maimonides rule and do not split into multiple classes. Despite this, we can consider this a setting for RDD where the treatment variable is whether or not the school is legally mandated to split a particular grade into smaller classes.

The dataset, published on the Harvard Dataverse \citep{ALHarvard2}, has verbal and math scores for 3rd, 4th and 5th graders across Israel. We chose to focus exclusively on 4th grade verbal scores as our response variable and 4th grade enrollments as our assignment variable because \cite{AngristLavy} suggest that a slightly significant effect of the treatment on 4th grade verbal scores exists. 
Even though the data were not 
generated by a tie-breaker we can still compute the
relative efficiency that a tie-breaker design would have had.

To simplify the analysis, we removed all schools that either had more than 80 students or more than two 4th-grade classes from the dataset. We further removed all schools that had NA entries for either class size or verbal scores, leaving $N=711$ schools in our filtered dataset. See Figure \ref{fig:hist_of_enrollments} for a visualization of the distribution of the 4th grade enrollments and Figure \ref{fig:RDD_fits_AL} for visualizations of the local linear regression based-RDD on this dataset using boxcar and triangular kernels.  We use the bandwidths $h_{\ik}$ given by the \cite{ImbensKalyanaraman_optimalBW} procedure, which were computed using that paper's MATLAB code.
The apparent benefit from smaller classrooms is positive but small and it turns
out, not statistically significant in this analysis.
The 95\% confidence interval (assuming homoscedastic errors) for the effect size at the boundary of the local linear regression-based RDD was $(-1.5,9.2)$ when a boxcar kernel with bandwidth $h_{\ik,\bc} = 7.09$ was used. The 95\% confidence interval for the effect size at boundary of this RDD was $(-2.4, 9.4)$ when a triangular kernel with bandwidth $h_{\ik,\ts} = 9.02$ was used.

\begin{figure}[t]
\centering
\includegraphics[width=0.75 \hsize]{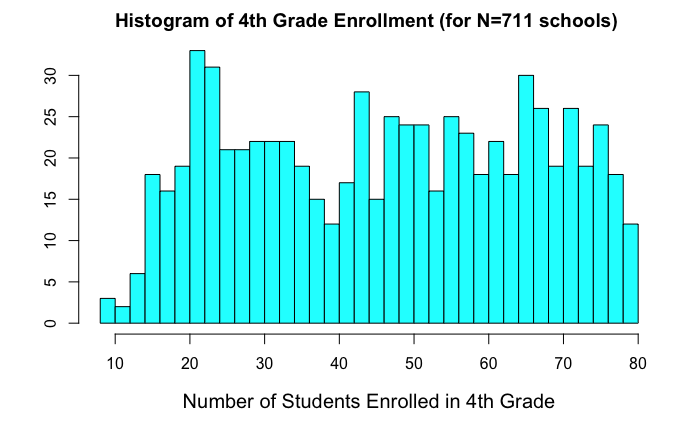}
\caption{\label{fig:hist_of_enrollments} A histogram of 4th grade enrollments for our filtered dataset (with schools exceeding 80 4th grade students or three 4th grade classes removed).}
\end{figure}

\begin{figure}[t]
    \centering
        \begin{subfigure}[b]{6 cm}
            \centering
            \includegraphics[width=6 cm]{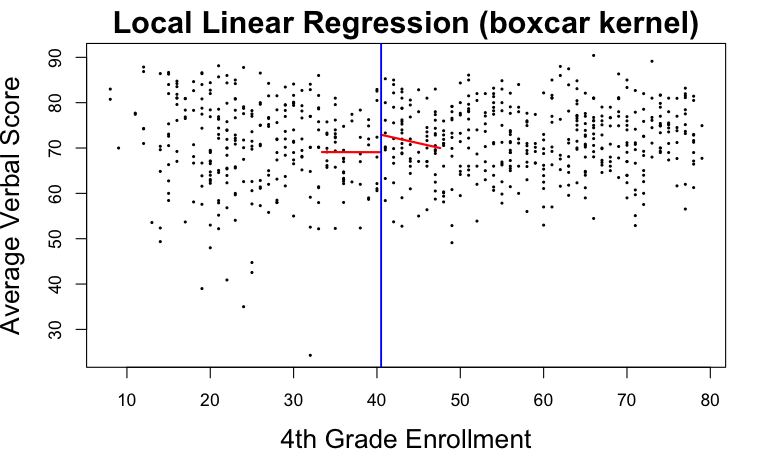}
            \label{fig:RDD_boxcarAL}
        \end{subfigure}
        \hfill
        \begin{subfigure}[b]{6 cm}
            \centering
            \includegraphics[width=6 cm]{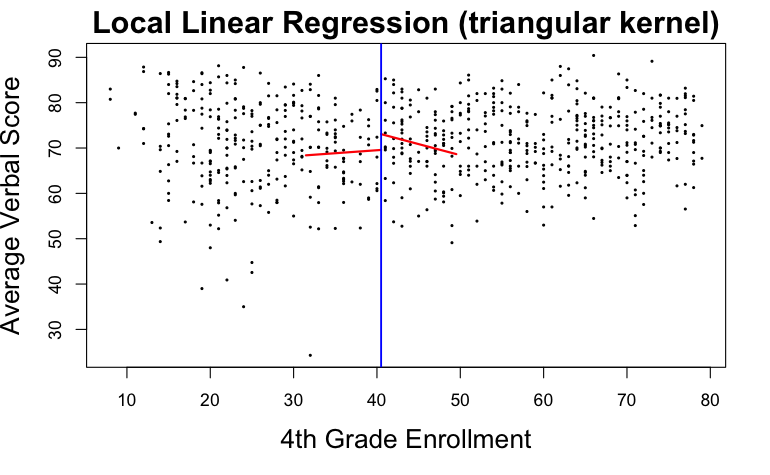}
            \label{fig:RDD_triangularAL}
        \end{subfigure}

        \caption[]
        {RDD fit to the 4th grader verbal scores from the \cite{ALHarvard2} dataset when using a boxcar kernel (left) and a triangular kernel (right). For these two fits, the bandwidths $h_{\ik,\bc}$ and $h_{\ik,\ts}$ were chosen as in \cite{ImbensKalyanaraman_optimalBW}. }
        \label{fig:RDD_fits_AL}
\end{figure}

Next we illustrate how an investigator can estimate the efficiency ratio of tie-breaker designs as a function of $\Delta$ on sample values of the assignment variable.
First we translate the data, replacing $x_i$ by $x_i-40.5$ to move the threshold from $t=40.5$ to $t=0$.
Next, for each $\Delta$ of interest we use $1000$ Monte Carlo samples to estimate $\var(\hat\beta_3\giv\cx;\Delta)$ and also $\var(\hat\beta_3\giv\cx ;0)$, both up to a constant $\sigma^2$.  That gives us $1000$ efficiency ratios $\eff^{(N)}(\Delta)=\var(\hat\beta_3\giv\cx;0)/\var(\hat\beta_3\giv\cx;\Delta)$ for each $\Delta$.
In each of our $1000$ samples, we simulate random assignments $Z_i$ for a tie-breaker design at the given experimental radius $\Delta$. The random assignments are stratified: in each consecutive pair of classroom sizes in the experimental region, one was randomly chosen to have $Z=1$ and the other got $Z=-1$.
The $x_i$ and the random $Z_i$ let us compute the matrices $\cx$ and $\cw$ defined in the beginning of Section \ref{sec:AsymApprox_via_integrals}, from which we compute a non-asymptotic $\var(\hat{\beta}_3\giv\cx;\Delta)$ using \eqref{eq:varbetahat}.
We do not simulate any $Y_i$ values because efficiency only depends on $\cx$, and
we are retaining the bandwidths from the \cite{ImbensKalyanaraman_optimalBW} procedure on the original data.  A more detailed simulation randomizing
the bandwidth choice is out of scope. Our simulations demonstrate that the TBD is more efficient at each fixed $h$, so we expect that it will also be more efficient at a randomly chosen $h$. There could be exceptions if the bandwidth is adversarially correlated with the estimation errors but we do not think
that is likely.

Figure \ref{fig:MC_boxcar_plots_AL} shows
boxplots of $1000$ simulated $\eff^{(N)}(\Delta)$ values
for various choices of $\Delta \in \mathbb{N}$ to plot the full efficiency curve. It is clear from Figure~\ref{fig:MC_boxcar_plots_AL} that with stratified allocations the efficiency is very reproducible.
Figure \ref{fig:Efficiency_plt_many_bandwidths} shows results for different bandwidths, ranging from $h_{\ik}/2$ to $3h_{\ik}/2$.  Because the efficiencies are so reproducible given the bandwidth, we just plot curves of the mean and standard deviations of estimated $\eff$ values. For both the boxcar and triangular kernels, we see that the tie-breaker design is reproducibly more efficient than the RDD and the effect increases as $\delta =\Delta/h$ increases for all $h$ we studied. The efficiency curves for this dataset under various bandwidth choices look similar to the theoretical efficiency curves derived in Section \ref{sec:AsymApprox_via_integrals} for the case of a uniform assignment variable.

\begin{figure}[t!]
    \centering
        \begin{subfigure}[b]{6 cm}
            \centering
            \includegraphics[width=6 cm]{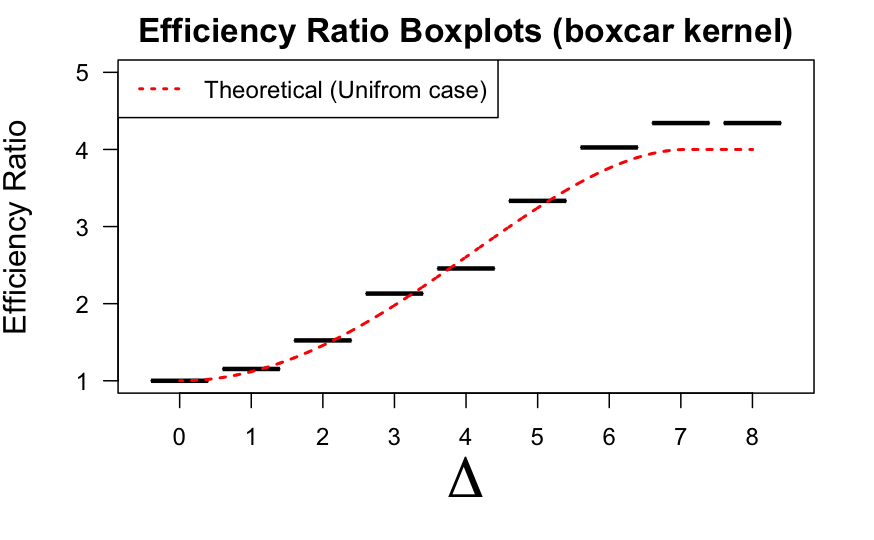}
            \label{fig:boxplot_boxcarAL}
        \end{subfigure}
        \hfill
        \begin{subfigure}[b]{6 cm}
            \centering
            \includegraphics[width=6 cm]{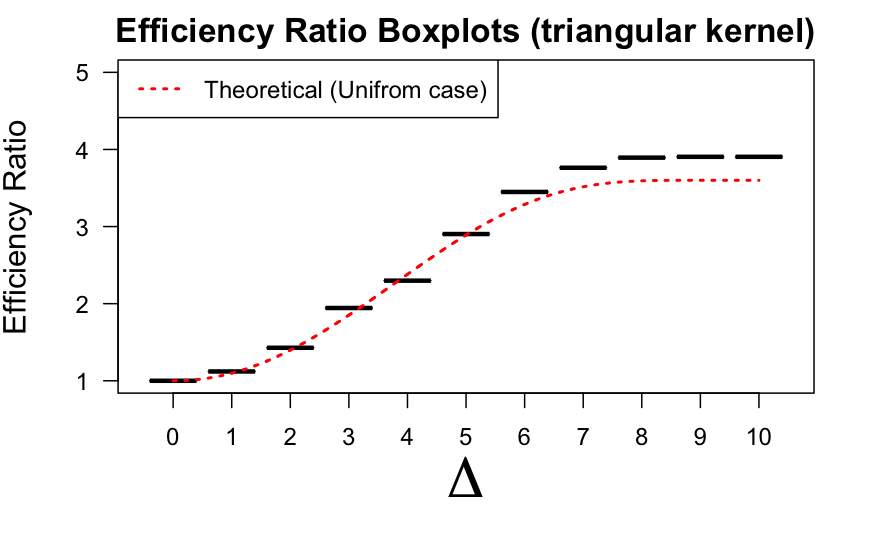}
            \label{fig:boxplot_triangularAL}
        \end{subfigure}

        \caption[]
        {Boxplots of the Monte-Carlo efficiency ratio estimates for various values of $\Delta \in \mathbb{N}$ when using a boxcar kernel (left) and a triangular kernel (right). For both kernels we used bandwidths from Figure \ref{fig:RDD_fits_AL}: $h_{\ik,\bc}=7.09$ and $h_{\ik,\ts}=9.02$. The dashed lines give the theoretical efficiency curves under the assumption of a uniform assignment variable, given by Propositions \ref{prop:boxcarefficiency} (left) and \ref{prop:triangleefficiency} (right). 
    \label{fig:MC_boxcar_plots_AL}}
\end{figure}

\begin{figure}[t]
    \centering
        \begin{subfigure}[b]{6 cm}
            \centering
            \includegraphics[width=6 cm]{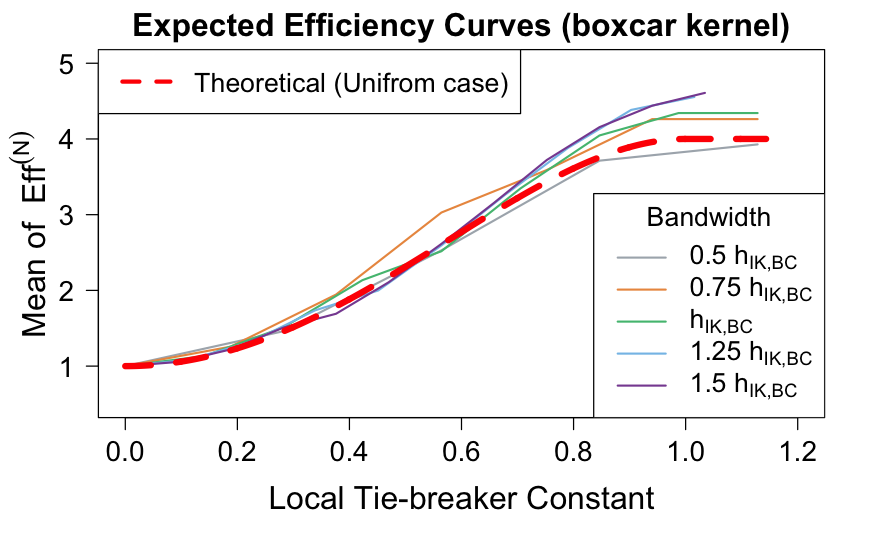}
            \label{fig:ExpectedEff_boxcar_AL2}
        \end{subfigure}
        \hfill
        \begin{subfigure}[b]{6 cm}
            \centering
            \includegraphics[width=6 cm]{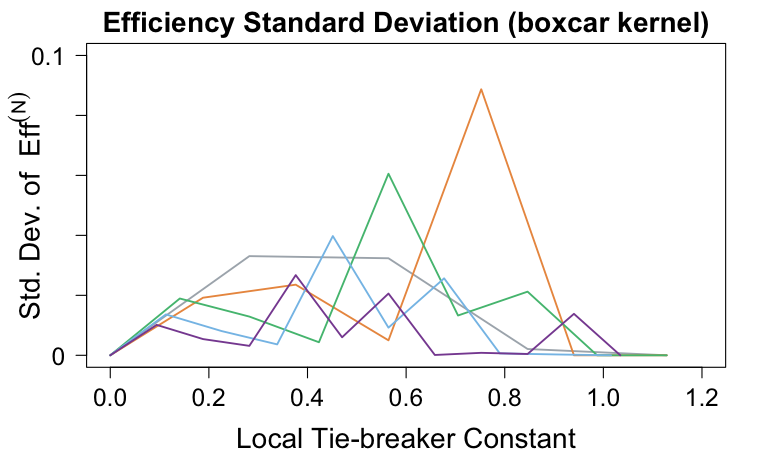}
            \label{fig:SD_Eff_AL_boxcar2}
        \end{subfigure}
        \begin{subfigure}[b]{6 cm}
            \centering
            \includegraphics[width=6 cm]{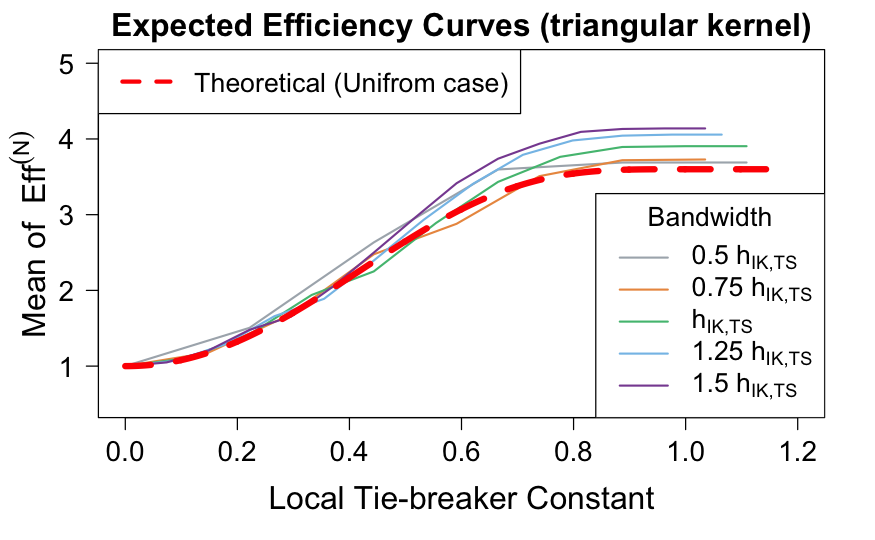}
            \label{fig:ExpectedEff_triangular_AL2}
        \end{subfigure}
        \hfill
        \begin{subfigure}[b]{6 cm}
            \centering
            \includegraphics[width=6 cm]{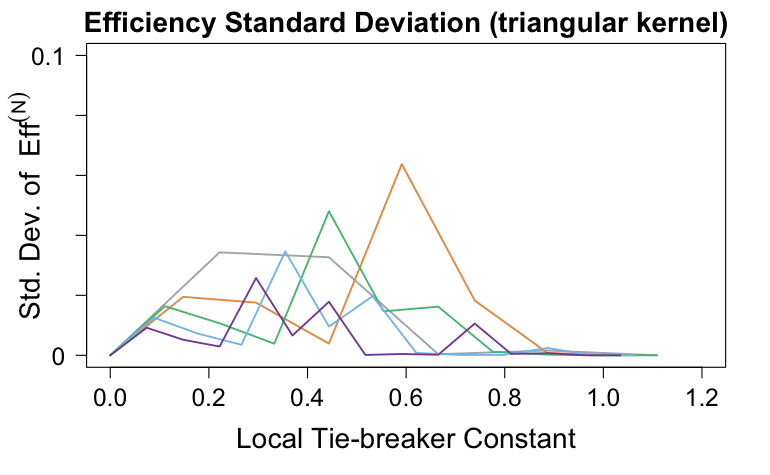}
            \label{fig:SD_Eff_AL_triangular2}
        \end{subfigure}
        \caption[]
        {Monte-Carlo estimates of the expected value (left) and standard deviation (right) of $\eff^{(N)}(\Delta)$ versus $\Delta/h$ for the \cite{ALHarvard2} dataset of 4th grader verbal scores. For these plots a boxcar kernel (top) and a triangular kernel (bottom) were used. The bandwidths plotted are scalar multiples of $h_{\ik,\bc}$ and $h_{\ik,\ts}$ from the procedure of \cite{ImbensKalyanaraman_optimalBW}. The legend for the plots on the right is the same as for the plots on the left. 
        These curves are not smooth because, to avoid redundancy, only points that corresponded to integer values of $\Delta$ were used. The thick dashed lines give the theoretical efficiency curves $\eff_{\bc}$ (top left) and $\eff_{\ts}$ (bottom left) derived in Section~\ref{sec:AsymApprox_via_integrals}.}
        \label{fig:Efficiency_plt_many_bandwidths}
\end{figure}

For a further discussion of the Maimonides rule, see
\cite{angr:etal:2019}.  They consider different data
sets and also investigate the possibility that the
class sizes are sometimes manipulated to be above
the threshold triggering a classroom split.

\subsection*{Comparison with theoretical results for uniform assignment variable}
Our theoretical analysis in Section \ref{sec:AsymApprox_via_integrals} is for a uniformly spaced assignment variable. We can offer one explanation for why the empirical efficiencies on non-uniformly distributed data look so similar to the theoretical ones for uniformly distributed data (see the left panels in Figure \ref{fig:Efficiency_plt_many_bandwidths}). The explanation uses some results about non-parametric regression from \citet[Table 2.1]{fan1996local}. Nonparametric regression estimates $\hat\mu(t)$ typically have an asymptotic variance where the leading term is proportional to $1/f(t)$ where $f$ is the probability density of the $x_i$. This arises because the local sample size is asymptotically proportional to $f(t)$. Hence, when considering nonuniform distributions, the $1/f(t)$ factors in the leading order variance terms will cancel out when computing the efficiency ratios. Some of the nonparametric regression estimators, such as the Nadaraya-Watson estimator, have a lead term in their bias that depends on the derivative $f'(t)$, and while $f'(t)=0$ for uniformly distributed data, it is not zero in general. Kernel weighted least squares methods (with symmetric $K(\cdot)$) do not have a dependency on $f'(t)$ in their bias. There is a curvature bias from $\mu''(t)$ but that is not related to the sampling distribution of the $x_i$. The lead terms in bias and variance for local linear regressions do not distinguish between distributions with the same value of $f(t)$ but different $f'(t)$. Thus the effects of non-uniformity of $X$ are asymptotically negligible.

\section{Discussion} \label{sec:Conclusion}

If an investigator is able to implement a 3-level tie-breaker design with any experimental radius $\Delta \in (0,\Delta_{\text{max}})$, our results show that the TBD has considerable statistical advantages over the RDD.

The most obvious advantage is that the TBD allows estimation of multiple causal parameters of interest including the average treatment effect over subjects with $x \in (t- \Delta,t+\Delta)$ as well as the expected treatment effect at any particular $x \in (t- \Delta,t+\Delta)$. The former is estimable at a faster rate and with fewer assumptions, whereas the latter may still be of interest for choosing a future policy threshold. Meanwhile, the RDD only allows estimation of $\tau_{\thresh}$, the expected treatment effect at $x=t$. 

Even if the only goal is estimation of $\tau_{\thresh}$, our results indicate a statistical advantage to running a TBD rather than an RDD and an advantage to picking a larger experimental radius $\Delta \in (0,\Delta_{\text{max}})$. As seen in Section \ref{sec:bandwidth_shrink_asymptotics}, to achieve the same asymptotic MSE in mean squared optimal estimation of $\tau_{\thresh}$, a TBD would require roughly 64 percent fewer samples than would be needed for an RDD. Moreover, the asymptotic advantage for a TBD is largely driven by its lower variance (Figure \ref{fig:bias_var_tradeoff}). 
Hence, if the convenient, but controversial, method of undersmoothing to construct asymptotically valid confidence intervals for $\tau_{\thresh}$ is used instead of more nearly optimal approaches, the TBD would exhibit even greater advantages over the RDD. We point readers to the introduction of \cite{Calonico_dont_undersmooth} for an overview of the history of undersmoothing, and \cite{calo:catt:farr:2019} for a modern approach to constructing confidence intervals that has better coverage properties than undersmoothing has.

In terms of the statistical advantages of picking a larger $\Delta$, \cite{owen:vari:2020} found an efficiency advantage for the tie-breaker in a global regression, wherein the estimation variance decreased monotonically in $\Delta$. We provide a comparable finding for the now more standard local linear regression approach: for any fixed bandwidth $h$, we see a theoretical efficiency that increases with the amount $\Delta$ of experimentation. We have not investigated the effect of $\Delta$ on the subsequent choice of $h$ when $\hat{h}_{\text{opt,TBD}} > \Delta$, although one candidate choice is an $h> \Delta$ that removes the leading order bias term, which we derived in Appendix~\ref{sec:MagicBandwidth}.

There is room for an improved estimator of $\gamma$
in the TBD context which uses data from both treatments on both
sides of the threshold $t$.  We leave this for further work.
A critical ingredient is the estimation of $\mu_{\pm}^{(2)}(t)$.
Compared to the method in \cite{ImbensKalyanaraman_optimalBW},
one could use a bandwidth tuned for an internal point $t$ instead
of one tuned for an endpoint.  Also the curvature estimates
in \cite{ImbensKalyanaraman_optimalBW}
use local quadratic regressions while \citet[p 63]{fan1996local}
suggest using local cubic regressions for curvature estimation at an interior point.

\begin{acks}[Acknowledgments]
This work was supported by the U.S.\ National Science
Foundation under grants IIS-1837931 
and DMS-2152780 and by Stanford University's SGF and SIGF fellowships.
We thank Hal Varian and Harrison Li for commenting on the paper as well as Steve Marron and Wolfgang H\"ardle for
some discussions about nonparametric regression.
We also thank anonymous reviewers for comments that
led us to improve the paper.
\end{acks}

\bibliographystyle{imsart-nameyear.bst}
\bibliography{TieBreakerWriteUp}

\appendix

\counterwithin{theorem}{section}
\counterwithin{lemma}{section}
\counterwithin{definition}{section}
\counterwithin{proposition}{section}
\counterwithin{remark}{section}
\counterwithin{corollary}{section}
\counterwithin{figure}{section}
\counterwithin{table}{section}

\section{Proof of Lemma~\ref{lemma:MSE_TBD}}\label{sec:proof_of_TBD_MSE}

Without loss of generality suppose $t=0$ (if this is not the case we can define a new assignment variable to be a translation of the original assignment variable by $t$). 
It is convenient to define 
$$\mathcal{Z}_+ \equiv \{ i\in\{1,2,\dots,N\} \ : \ Z_i=1 \} \quad \text{ and } \quad \mathcal{Z}_- \equiv \{ i\in\{1,2,\dots,N\}: \ Z_i=-1 \}.$$ 
Also define $N_\pm \equiv \vert \mathcal{Z}_\pm \vert$. 
Letting $K_h(u) \equiv K(u/h)/h$, define 
$$F_{j,\pm} \equiv \frac{1}{N_\pm} \sum_{i \in \mathcal{Z}_\pm} X_i^j  K_h(X_i) \quad \text{ and } \quad
G_{j,\pm} \equiv \frac{1}{N_\pm} \sum_{i \in \mathcal{Z}_\pm} X_i^j K_h^2(X_i) \sigma_\pm^2(X_i),
$$ 
for nonnegative integers $j$. 

Next, let $f_\pm$ be the conditional density of $X_i$ given that $Z_i=\pm1$. 
That is 
\begin{equation}\label{eq:fpm_def}
    f_+(x) = \begin{cases} 
\frac{f(x)}{\Pr(Z_i=1)},  &  x \geq \Delta \\ \frac{f(x)/2}{ \Pr(Z_i=1)},  &  |x|<\Delta\\  0, & 
x \leq - \Delta \end{cases} 
\quad  \text{and} \quad  
f_-(x) = \begin{cases} 0,  & x \geq \Delta \\ \frac{f(x)/2}{ \Pr(Z_i=-1)},  & |x|<\Delta \\  \frac{f(x)}{\Pr(Z_i=-1)}, & x \leq - \Delta. \end{cases}  
\end{equation}
We now prove some helpful small bandwidth approximations for $F_{j,\pm}$
and $G_{j,\pm}$.

\begin{lemma}\label{lemma:F_formula} In the setting of Lemma \ref{lemma:MSE_TBD}, for nonnegative integers $j$, $$F_{j,\pm}=h^jf_\pm(0)\nu_j +o_p(h^j)$$ 
with $\nu_j=\int_{-\infty}^{\infty} u^j K(u)\rd u $.
\end{lemma}

\begin{proof}

We will prove this for $F_{j,+}$ and the proof for $F_{j,-}$ will be identical. First note that $F_{j,+}$ is the average of $N_+$ IID random variables so
$$F_{j,+}=\e[X^j K_h(X)|Z=1]+O_p\Big( \sqrt{\text{Var}[X^j K_h(X)|Z=1] /N_+} \Big).$$
Now note that
$$ \begin{aligned} \e[X^j K_h(X)|Z=1]  &= \int_{-\infty}^{\infty} \frac{1}{h} K(x/h) x^j f_+(x) \rd x \\
 &=  h^j \int_{-\infty}^{\infty}  K(u) u^j f_+(uh) \rd u \\ 
  &=  h^j \int_{-\infty}^{\infty}  K(u) u^j f_+(0) \rd u +  h^{j+1} \int_{-\infty}^{\infty}  K(u) u^{j+1} \frac{f_+(hu)-f_+(0)}{hu} \rd u \\ 
   &= h^j f_+(0) \nu_j +h^{j+1} \int_{-\infty}^{\infty}  K(u) u^{j+1} \frac{f_+(hu)-f_+(0)}{hu} \rd u \\ 
    &= h^j f_+(0) \nu_j + O(h^{j+1}), \end{aligned}$$ 
    where the last equality holds because $K(\cdot)$ is bounded with bounded support, $f_+$ is continuously differentiable at $0$, and asymptotically $h \to 0$.

Meanwhile, the term $\text{Var}[X^j K_h(X)|Z=1]$ is upper bounded by 
$$\begin{aligned}  \e[(K_h(X))^2 X^{2j}|Z=1] & =  \frac{1}{h^2} \int_{-\infty}^{\infty} \big( K(x/h) \big)^2 x^{2j} f_+(x) \rd x \\ 
& =  h^{2j-1} \int_{-\infty}^{\infty}  K^2(u) u^{2j} f_+(hu) \rd u \\ 
& =  O(h^{2j-1}),  \end{aligned} $$ where the last equality holds because $K(\cdot)$ is bounded with bounded support, $f_+$ is continuous and strictly positive at $0$, and asymptotically $h \to 0$. Thus 
\begin{align*}\frac{\text{Var}[X^j K_h(X)|Z=1]}{N_+} &\leq \frac{ \e[(K_h(X_i))^2 X_i^{2j}|Z=1]}{N_+} \\
&= \frac{ O(h^{2j-1})}{N_+} =O\Bigl(\frac{h^{2j}}{hN}\Bigr)\cdot \frac{N}{N_+}=o_p(h^{2j}),
\end{align*}
where last equality holds because as $N \to \infty$, $hN \to \infty$, so $h^{2j}/hN =o(h^{2j})$, and because $N/N_+=O_p(1)$. 

Combining previous results $$F_{j,+} =h^j f_+(0) \nu_j+O(h^{j+1})+O_p(o_p(h^{j}))=h^j f_+(0) \nu_j +o_p(h^j).\qedhere$$
\end{proof}

\begin{lemma}\label{lemma:G_formula} In the setting of Lemma \ref{lemma:MSE_TBD}, for nonnegative integers $j$, 
$$G_{j,\pm}=h^{j-1} \sigma_\pm^2(0)f_\pm(0)\pi_j (1 +o_p(1))$$ 
with $\pi_j=\int_{-\infty}^{\infty} u^j K^2(u)\rd u $.
\end{lemma}

\begin{proof}

We will prove this for $G_{j,+}$ and the proof for $G_{j,-}$ will be identical. First note that $G_{j,+}$ is the average of $N_+$ IID random variables so 
$$G_{j,+}=\e[X^j K_h^2(X) \sigma_+^2(X)|Z=1]+O_p\Big( \sqrt{\text{Var}[X^j K_h^2(X) \sigma_+^2(X)|Z=1] /N_+} \Big).$$
Now note that
$\e[X^j K_h^2(X)\sigma_+^2(X)|Z=1]$ equals
 \begin{align}\label{eq:diff_sigma_needed} 
 &  \frac{1}{h^2} \int_{-\infty}^{\infty} K^2(x/h) x^j \sigma_+^2(x) f_+(x) \rd x \nonumber \\  
 &=  h^{j-1} \int_{-\infty}^{\infty} K^2(u) u^j \sigma_+^2(hu) f_+(hu) \rd u \nonumber \\ 
  &= h^{j-1} \pi_j \sigma_+^2(0) f_+(0) + h^{j-1} \int_{-\infty}^{\infty} K^2(u) u^{j} \big[ \sigma_+^2(hu) f_+(hu) - \sigma_+^2(0)f_+(0) \big] \rd u \nonumber \\ 
   &= h^{j-1} \pi_j \sigma_+^2(0) f_+(0) + o(h^{j-1}), 
   \end{align}
   where the last equality holds because $K(\cdot)$ is bounded with bounded support, $ \sigma_+^2(\cdot) f_+(\cdot)$ is continuous at $0$, and asymptotically $h \to 0$. 

Meanwhile, the term $\text{Var}[X^j K_h^2(X) \sigma_+^2(X)|Z=1]$ is upper bounded by
$$\begin{aligned}  \e[(K_h(X))^4 X^{2j} \sigma_+^4(X)|Z=1] & =  \frac{1}{h^4} \int_{-\infty}^{\infty} K^4(x/h) x^{2j} \sigma_+^4(x) f_+(x)\rd x \\ 
& =  h^{2j-3} \int_{-\infty}^{\infty} K^4(u) u^{2j} \sigma_+^4(hu) f_+(hu)\rd u \\ 
&=  O(h^{2j-3}),  \end{aligned} $$
where the last equality holds because $K(\cdot)$ is bounded with bounded support, $ \sigma_+^2(\cdot) f_+(\cdot)$ is continuous and strictly positive at $0$, and asymptotically $h \to 0$.

Thus
$$\begin{aligned} \frac{\text{Var}[X^j K_h^2(X) \sigma_+^2(X)|Z=1]}{N_+} & \leq \frac{ \e[X^j K_h^2(X)\sigma_+^2(X)|Z=1]}{N_+} \\ 
& =  \frac{ O(h^{2j-3})}{N_+} = O(\frac{h^{2j-2}}{hN})\cdot \frac{N}{N_+} = o_p(h^{2j-2}), \end{aligned}$$ 
where the last equality holds because as $N \to \infty$, $hN \to \infty$, so $h^{2j}/hN =o(h^{2j})$, and because $N/N_+=O_p(1)$. 
Combining previous results, 
\begin{align*}
G_{j,+} &= h^{j-1} \pi_j \sigma_+^2(0) f_+(0) + o(h^{j-1})+O_p(o_p(h^{j-1}))\\
&=h^{j-1} \pi_j \sigma_+^2(0) f_+(0) +o_p(h^{j-1}).\qedhere
\end{align*}
\end{proof}

We are now ready to compute an asymptotic approximation for the bias and variance of the causal estimator $\hat{\tau}_{\thresh} =\hat{\tau}(0)$. Recall that as discussed in Section \ref{sec:CausalParamsAndProblem}, $\hat{\tau}(0)$ can be equivalently estimated by solving two separate local linear regressions, one for the treatment group and one for the control group, rather than solving \eqref{eq:kernelcriterion} and plugging the solution into \eqref{eq:def_hat_tauthresh}. In this proof, we use the two separate local linear regression formulation, so that the proof more closely resembles that seen in the appendix of \cite{ImbensKalyanaraman_optimalBW}.

In particular, define $\mathcal{X}_+ \in \mathbb{R}^{N_+ \times 2} $ and $\mathcal{X}_-  \in \mathbb{R}^{N_- \times 2} $ to be the design matrices for the local linear regression restricted to the treated group and the control group respectively. That is define $$\mathcal{X}_+ \equiv \begin{bmatrix} \textbf{1} & (X_i)_{i \in \mathcal{Z}_+} \end{bmatrix} \quad \text{ and } \quad \mathcal{X}_- \equiv \begin{bmatrix} \textbf{1} & (X_i)_{i \in \mathcal{Z}_-} \end{bmatrix}.$$ Also define the corresponding local linear regression weight matrices by $$W_+ \equiv \text{diag} \Big( 
\big(K_h(X_i)\big)_{i \in \mathcal{Z}_+} \Big) \quad \text{ and } \quad W_- \equiv \text{diag} \Big( 
\big(K_h(X_i)\big)_{i \in \mathcal{Z}_-} \Big),$$ and similarly define the corresponding conditional variance matrices by $$\Sigma_+ \equiv \text{diag} \Big( 
\big(\sigma_+^2(X_i)\big)_{i \in \mathcal{Z}_+} \Big) \quad \text{ and } \quad \Sigma_- \equiv \text{diag} \Big( 
\big(\sigma_-^2(X_i)\big)_{i \in \mathcal{Z}_-} \Big).$$ Finally, define $\mathcal{Y}_+ \equiv (Y_i)_{i \in \mathcal{Z}_+}$, $\mathcal{Y}_- \equiv (Y_i)_{i \in \mathcal{Z}_-}$, and let $\mathbf{e}_1=(1,0)$. The causal estimator for $\tau(0)=\mu_+(0)-\mu_-(0)$ is given by $\hat{\tau}(0)=\hat{\mu}_+(0)-\hat{\mu}_-(0)$, where $\hat{\mu}_+(0)$ and $\hat{\mu}_-(0)$ are the local linear regression estimators given by $$\hat{\mu}_+(0)=\mathbf{e}_1^{\tran} (\mathcal{X}_+^{\tran} W_+ \mathcal{X}_+)^{-1} \mathcal{X}_+^{\tran} W_+ \mathcal{Y}_+ \ \ \ \text{and} \ \ \  \hat{\mu}_-(0)=\mathbf{e}_1^{\tran} (\mathcal{X}_-^{\tran} W_- \mathcal{X}_-)^{-1} \mathcal{X}_-^{\tran} W_- \mathcal{Y}_-.$$

Let $\mathcal{X} \in \mathbb{R}^{N \times 4}$ be the full design matrix whose ith row is $(1,X_i,Z_i,X_iZ_i)$, and note that the matrices $\mathcal{X}_+$ and $\mathcal{X}_-$ and the sets $\mathcal{Z}_+$ and $\mathcal{Z}_-$ are functions of the full design matrix.

In the next subsection, we compute the asymptotic approximation to the bias of the estimator $\e[\hat{\tau}(0)| \mathcal{X}]-\tau(0)$, and in the following subsection, we compute the asymptotic approximation to its variance $\text{Var}[\hat{\tau}(0)| \mathcal{X}]$. These calculations leverage Lemmas \ref{lemma:F_formula} and \ref{lemma:G_formula}.

\subsection*{Asymptotic approximation of the bias}

Define $B_+ \equiv \e[\hat{\mu}_+(0)\giv \mathcal{X}]-\mu_+(0)$ and $B_- \equiv \e[\hat{\mu}_-(0)| \mathcal{X}]-\mu_-(0)$, and note that $$\e[\hat{\tau}(0)| \mathcal{X}]-\tau(0)=B_+ -B_-.$$ 
We will compute the asymptotic formula for $B_+$, and by an identical argument, the asymptotic formula for $B_-$ will follow.
 
To do this note that $$\e[\hat{\mu}_+(0)| \mathcal{X}] = \mathbf{e}_1^{\tran} (\mathcal{X}_+^{\tran} W_+ \mathcal{X}_+)^{-1} \mathcal{X}_+^{\tran} W_+ \e[\mathcal{Y}_+| \mathcal{X}] =\mathbf{e}_1^{\tran} (\mathcal{X}_+^{\tran} W_+ \mathcal{X}_+)^{-1} \mathcal{X}_+^{\tran} W_+ M_+$$ where $M_+=\big( \mu_+(X_i) \big)_{i \in \mathcal{Z}_+}$.
 
 Since $\mu_+(\cdot)$ has at least three continuous derivatives in an open neighborhood of $0$, for each $i \in \mathcal{Z}_i$,  $$\mu_+(X_i) = \mu_+(0) +\mu_+^{(1)}(0) X_i + \frac{1}{2} \mu_+^{(2)}(0) X_i^2 +T_i,$$ where $\vert T_i \vert \leq  \vert X_i^3 \vert \cdot \sup_{  x \in [- \vert X_i \vert,\vert X_i \vert] }  \vert \mu_+^{(3)}(x) \vert $.
 
 So letting $T_+=\big(T_i \big)_{i \in \mathcal{Z}_+}$ and $S_+=\big(\mu_+^{(2)}(0) X_i^2/2 \big)_{i \in \mathcal{Z}_+}$ it follows that $$M_+=\mathcal{X}_+ \begin{bmatrix} \mu_+(0) \\ \mu_+^{(1)} (0) \end{bmatrix} +S_++T_+.$$ Combining previous results 
 \begin{align*}B_+&=e_1^{\tran} (\mathcal{X}_+^{\tran} W_+ \mathcal{X}_+)^{-1} \mathcal{X}_+^{\tran} W_+ M_+ -\mu_+(0) \\
& =e_1^{\tran} (\mathcal{X}_+^{\tran} W_+ \mathcal{X}_+)^{-1}  \mathcal{X}_+^{\tran} W_+ (S_+ + T_+).
 \end{align*}

Since $\nu_0\nu_2 -\nu_1^2>0$ (by the Cauchy-Schwartz inequality),  Lemma \ref{lemma:F_formula} and a first order Taylor expansion, as $h \to 0$ and $N \to \infty$, yield $$\frac{1}{F_{0,+}F_{2,+} - F_{1,+}^2}= \frac{1}{h^2 f_+^2(0)[ \nu_0 \nu_2 - \nu_1^2] } \Big( \frac{1}{1+o_p(1)}  \Big) = \frac{1 +o_p(1)}{h^2 f_+^2(0)[ \nu_0 \nu_2 - \nu_1^2] }. $$

 Hence, by Lemma \ref{lemma:F_formula} again,  
 $$\begin{aligned}\Big( \frac{1}{N_+} \mathcal{X}_+^{\tran} W_+ \mathcal{X}_+ \Big)^{-1}    
 & =  \begin{bmatrix} F_{0,+} & F_{1,+} \\ F_{1,+} & F_{2,+} \end{bmatrix}^{-1} \\ 
 & =  \frac{1}{F_{0,+}F_{2,+} - F_{1,+}^2} \begin{bmatrix} F_{2,+} & -F_{1,+} \\ -F_{1,+} & F_{0,+} \end{bmatrix} \\ 
 & =   \big( 1+ o_p(1) \big)\begin{bmatrix} \frac{h^2 [ f_+(0) \nu_2+ o_p(1)]}{h^2 f_+^2(0) [\nu_0 \nu_2 - \nu_1^2]} & \frac{-h [ f_+(0) \nu_1+ o_p(1)] }{h^2 f_+^2(0) [\nu_0 \nu_2 - \nu_1^2]} \\ \\ \frac{-h [ f_+(0) \nu_1+ o_p(1)]}{h^2 f_+^2(0) [\nu_0 \nu_2 - \nu_1^2]} & \frac{f_+(0) \nu_0+ o_p(1)}{h^2 f_+^2(0) [\nu_0 \nu_2 - \nu_1^2]} \end{bmatrix}  \\  
 & = \begin{bmatrix} \frac{\nu_2}{f_+(0) ( \nu_0 \nu_2 -\nu_1^2)}+ o_p(1) & \frac{-\nu_1}{h f_+(0) ( \nu_0 \nu_2 -\nu_1^2)}+ o_p(1/h) \\ \frac{-\nu_1}{h f_+(0) ( \nu_0 \nu_2 -\nu_1^2)}+ o_p(1/h) & \frac{\nu_0}{h^2f_+(0) ( \nu_0 \nu_2 -\nu_1^2)}+ o_p(1/h^2) \end{bmatrix}   \\ 
 & =  \begin{bmatrix} O_p(1) & O_p(1/h) \\ O_p(1/h) & O_p(1/h^2) \end{bmatrix}. \end{aligned} $$

 Now note, that since $\mu_+^{(3)}(\cdot)$ is continuous in an open neighborhood of $0$, there exists an $a>0$ for which $\sup_{x \in [-a,a]} \vert \mu_+^{(3)} (x) \vert \equiv \overline{a}_+ < \infty$. Since $K(\cdot)$ has bounded support, there exists an $\epsilon>0$ such that for all $h \in (0,\epsilon)$, $K_h(x)=0$ whenever $\vert x \vert > a$. Therefore for all $h < \epsilon$ and all $i$, $$\vert K_h(X_i) T_i \vert \leq \Big| K_h(X_i) \cdot \vert X_i^3 \vert \cdot \sup_{  x \in [- \vert X_i \vert,\vert X_i \vert] }  \vert \mu_+^{(3)}(x) \vert  \Big|  \leq \overline{a}_+ K_h(X_i) \vert X_i^3 \vert.$$ Thus for all $h$ sufficiently small, $$ \frac{1}{N_+} \mathcal{X}_+^{\tran} W_+ T_+  \leq \overline{a}_+ \begin{bmatrix} \frac{1}{N_+} \sum_{i \in \mathcal{Z}_+ } K_h(X_i) \vert X_i^3 \vert \\ F_{4,+} \end{bmatrix} = \begin{bmatrix} O_p(h^3) \\ O_p(h^4) \end{bmatrix} $$ where the bottom equality holds by Lemma \ref{lemma:F_formula}, and the top equality holds by the exact same argument as the proof of Lemma \ref{lemma:F_formula} except with absolute values.
 
 Combining the two previous results and multiplying through and dividing by $N_+$, $$e_1^{\tran} (\mathcal{X}_+^{\tran} W_+ \mathcal{X}_+)^{-1}  \mathcal{X}_+^{\tran} W_+T_+ =O_p(1) O_p(h^3) +O_p(1/h) O_p(h^4)=O_p(h^3)=o_p(h^2).$$
 
 Similarly, by Lemma \ref{lemma:F_formula}, $$ \frac{1}{N_+}\mathcal{X}_+^{\tran} W_+S_+ = \frac{\mu_+^{(2)} (0)}{2}  \begin{bmatrix} F_{2,+} \\ F_{3,+} \end{bmatrix}= \frac{\mu_+^{(2)} (0)}{2}   f_+(0) \begin{bmatrix} \nu_2 h^2 +o_p(h^2) \\ \nu_3 h^3 + o_p(h^3) \end{bmatrix} ,$$ 
 and therefore, 
 $$\begin{aligned} B_+ =  & e_1^{\tran} (\mathcal{X}_+^{\tran} W_+ \mathcal{X}_+)^{-1}  \mathcal{X}_+^{\tran} W_+ (S_+ + T_+) \\ 
 & =  e_1^{\tran} (\frac{1}{N_+} \mathcal{X}_+^{\tran} W_+ \mathcal{X}_+)^{-1} (\frac{1}{N_+} \mathcal{X}_+^{\tran} W_+ S_+) + o_p(h^2) \\ 
 & = \frac{1}{2} \mu_+^{(2)}(0) f_+(0) \Big[[\frac{\nu_2}{f_+(0) ( \nu_0 \nu_2 -\nu_1^2)}+ o_p(1)] [\nu_2 h^2 +o_p(h^2)]  \Big] \\ 
  & \phe+ \frac{1}{2} \mu_+^{(2)}(0) f_+(0) \Big[[-\frac{\nu_1}{h f_+(0) ( \nu_0 \nu_2 -\nu_1^2)}+ o_p(1/h)] [\nu_3 h^3 + o_p(h^3)]  \Big] + o_p(h^2) \\ 
 & = \frac{1}{2} \mu_+^{(2)}(0) \Big( \frac{\nu_2^2 - \nu_3 \nu_1}{\nu_0 \nu_2 - \nu_1^2} \Big)h^2 + o_p(h^2). \end{aligned}$$
A similar argument shows that $$B_-= \frac{1}{2} \mu_-^{(2)}(0) \Big( \frac{\nu_2^2 - \nu_3 \nu_1}{\nu_0 \nu_2 - \nu_1^2} \Big)h^2 + o_p(h^2).$$ Hence, recalling that the bias is given by $\e[\hat{\tau}(0)| \mathcal{X}]-\tau(0)=B_+ -B_-$ it follows that the asymptotic formula for the bias is \begin{equation}\label{eq:as_bias_TBD} \text{bias}_{\hat{\tau}(0)}=
     \e[\hat{\tau}(0)| \mathcal{X}]-\tau(0) = \frac{1}{2} \Big( \mu_+^{(2)}(0) - \mu_-^{(2)}(0) \Big) \Big( \frac{\nu_2^2 - \nu_3 \nu_1}{\nu_0 \nu_2 - \nu_1^2} \Big)h^2 + o_p(h^2).
 \end{equation} 
 Squaring the above formula, we recover the leading order squared-bias. As noted in Section \ref{sec:bandwidth_shrink_asymptotics}, the leading order squared-bias is the first term in \eqref{eq:AMSE_TBD_def}.

\subsection*{Asymptotic approximation of the variance}

Defining $V_\pm= \text{Var}(\hat{\mu}_\pm(0) \giv \mathcal{X})$, 
the variance for our causal estimator is  $$\text{Var}[\hat{\tau}(0)\giv \mathcal{X}] 
= V_++V_- +2 \text{Cov} [\hat{\mu}_+(0),\hat{\mu}_-(0)\giv\mathcal{X} ] = V_++V_- ,$$ where the last equality holds because $(X_i,Y_{i+},Y_{i-})$ are IID by assumption making $\mathcal{Y}_+$ and $\mathcal{Y}_-$  independent conditionally on the treatment assignments which
then implies that $\hat{\mu}_\pm(0)$ 
are independent conditionally on $\mathcal{X}$.

We now compute an asymptotic formula for $V_+$ and by an identical argument, the asymptotic formula for $V_-$ will follow.

First note that 
$$\begin{aligned}V_+  &= \text{Var}[ \mathbf{e}_1^{\tran} (\mathcal{X}_+^{\tran} W_+ \mathcal{X}_+)^{-1} \mathcal{X}_+^{\tran} W_+ \mathcal{Y}_+ | \mathcal{X}] \\ 
 &= \mathbf{e}_1^{\tran} (\mathcal{X}_+^{\tran} W_+ \mathcal{X}_+)^{-1} \mathcal{X}_+^{\tran} W_+ \text{Var}[\mathcal{Y}_+ | \mathcal{X}] W_+^{\tran} \mathcal{X}_+ (\mathcal{X}_+^{\tran} W_+ \mathcal{X}_+)^{-1} \mathbf{e}_1 \\ 
  &= \mathbf{e}_1^{\tran} (\mathcal{X}_+^{\tran} W_+ \mathcal{X}_+)^{-1} \mathcal{X}_+^{\tran} W_+ \Sigma_+ W_+ \mathcal{X}_+ (\mathcal{X}_+^{\tran} W_+ \mathcal{X}_+)^{-1} \mathbf{e}_1  \end{aligned}$$
Now the middle factor rescaled by $1/N_+$ is 
$$ \frac{1}{N_+}  \mathcal{X}_+^{\tran} W_+ \Sigma_+ W_+ \mathcal{X}_+ = \begin{bmatrix} G_{0,+} & G_{1,+} \\ G_{1,+} & G_{2,+} \end{bmatrix}$$ and recall from the previous subsection that $$\Big( \frac{1}{N_+} \mathcal{X}_+^{\tran} W_+ \mathcal{X}_+ \Big)^{-1} =  \frac{1}{F_{0,+}F_{2,+} - F_{1,+}^2} \begin{bmatrix} F_{2,+} & -F_{1,+} \\ -F_{1,+} & F_{0,+} \end{bmatrix}.$$ 
Thus, 
$$\begin{aligned} N_+ V_+  &= \frac{1}{(F_{0,+}F_{2,+} - F_{1,+}^2)^2} \mathbf{e}_1^{\tran}  \begin{bmatrix} F_{2,+} & -F_{1,+} \\ -F_{1,+} & F_{0,+} \end{bmatrix} \begin{bmatrix} G_{0,+} & G_{1,+} \\ 
G_{1,+} & G_{2,+} \end{bmatrix} \begin{bmatrix} F_{2,+} & -F_{1,+} \\ -F_{1,+} & F_{0,+} \end{bmatrix} \mathbf{e}_1 \\ 
 &= \frac{F_{2,+}^2 G_{0,+} -2F_{1,+} F_{2,+} G_{1,+} +F_{1,+}^2 G_{2,+} }{(F_{0,+}F_{2,+} -F_{1,+}^2)^2 } \\ 
  &=   \frac{\sigma_+^2(0)}{f_+(0)} \cdot \frac{ \Big(\nu_2^2 \pi_0 h^3 - 2 \nu_1 \nu_2 \pi_1 h^3 + \nu_1^2 \pi_2 h^3 \Big)[1+o_p(1)]^3}{ (\nu_0 \nu_2 h^2 - \nu_1^2 h^2)^2 [1+o_p(1)]^4} \\ 
 &= \frac{\sigma_+^2(0)}{h f_+(0)} \Big( \frac{\nu_2^2 \pi_0 - 2\nu_1 \nu_2 \pi_1 + \nu_1^2 \pi_2}{ (\nu_0 \nu_2 - \nu_1^2)^2 } +o_p(1) \Big) \\ 
  &= \frac{2 \sigma_+^2(0) \Pr(Z_i=1)}{h f(0)} \Big( \frac{\nu_2^2 \pi_0 - 2\nu_1 \nu_2 \pi_1 + \nu_1^2 \pi_2}{ (\nu_0 \nu_2 - \nu_1^2)^2 } +o_p(1) \Big).  
\end{aligned}  $$

Thus dividing through by $N_+$ and rearranging terms, $$V_+= \frac{2 \sigma_+^2(0) }{N h f(0)} \Big( \frac{\nu_2^2 \pi_0 - 2\nu_1 \nu_2 \pi_1 + \nu_1^2 \pi_2}{ (\nu_0 \nu_2 - \nu_1^2)^2 } +o_p(1) \Big) \cdot \frac{N \Pr(Z_i=1)}{N_+}. $$

By the weak law of large numbers $N \Pr(Z_i=1)/N_+ =(1+o_p(1))$ and therefore, 
$$V_+= \frac{2\sigma_+^2(0) }{N h f(0)} \Big( \frac{\nu_2^2 \pi_0 - 2\nu_1 \nu_2 \pi_1 + \nu_1^2 \pi_2}{ (\nu_0 \nu_2 - \nu_1^2)^2 }  \Big) +o_p\Bigl(\frac{1}{Nh}\Bigr).$$ 
A similar argument shows that $$V_-= \frac{2\sigma_-^2(0)  }{N h f(0)} \Big( \frac{\nu_2^2 \pi_0 - 2\nu_1 \nu_2 \pi_1 + \nu_1^2 \pi_2}{ (\nu_0 \nu_2 - \nu_1^2)^2 }  \Big) +o_p\Bigl(\frac{1}{Nh}\Bigr),$$ 
and hence because $\text{Var}[\hat{\tau}(0)| \mathcal{X}]=V_++V_-$,  \begin{equation}\label{eq:TBD_as_var_formula} \text{Var}[\hat{\tau}(0)| \mathcal{X}] = \frac{\sigma_+^2(0) +\sigma_-^2(0)}{N h f(0)} \Big( \frac{2\nu_2^2 \pi_0 - 4\nu_1 \nu_2 \pi_1 + 2\nu_1^2 \pi_2}{ (\nu_0 \nu_2 - \nu_1^2)^2 }  \Big) +o_p\Bigl(\frac{1}{Nh}\Bigr). \end{equation}
As remarked in Section \ref{sec:bandwidth_shrink_asymptotics}, the leading order variance matches the second term in formula \eqref{eq:AMSE_TBD_def}.

\subsection*{Asymptotic expression for mean squared error}

To complete the proof of the asymptotic expression for the tie-breaker design MSE in Lemma \ref{lemma:MSE_TBD}, note that $$\begin{aligned} \text{MSE}_{\text{TBD}}(h,N) & = \e \big[ \big( \hat{\tau}(0)-\tau(0) \big)^2 \big| \mathcal{X} \big] \\ & = \text{bias}_{\hat{\tau}(0)}^2+\text{Var}[\hat{\tau}(0)| \mathcal{X}] \\ & = C_1 \big( \mu_+^{(2)}(0) - \mu_-^{(2)}(0) \big)^2  h^4 + o_p(h^4) + \frac{C_2  }{N h } \Big( \frac{ \sigma_+^2(0) +  \sigma_-^2(0)}{f(0)}  \Big) +o_p\Bigl(\frac{1}{Nh}\Bigr) \\ & = \text{AMSE}_{\text{TBD}}(h,N)+o_p\Bigl(h^4+\frac{1}{Nh}\Bigr)  \end{aligned}$$ 
where the second last equality holds from combining equations \eqref{eq:TBD_as_var_formula}, \eqref{eq:as_bias_TBD}, and \eqref{eq:C_def}, while the last equality holds from definition \eqref{eq:AMSE_TBD_def}.

\subsection*{Asymptotically optimal bandwidth expression}

To complete the proof of Lemma \ref{lemma:MSE_TBD}, recall that we define 
$$h_{\text{opt,TBD}}(N)= \argmin_h \text{AMSE}_{\text{TBD}}(h,N) $$ 
and by \eqref{eq:AMSE_TBD_def},  $$\text{AMSE}_{\text{TBD}}(h,N) \equiv C_1 \big( \mu_+^{(2)}(0) -\mu_-^{(2)}(0) \big)^2 h^4+ \frac{C_2 }{Nh} \Big( \frac{  \sigma_+^2(0)+\sigma_-^2(0)}{f(0)} \Big).$$

Because by assumption $\mu_+^{(2)}(0) \neq \mu_-^{(2)}(0)$ and $f(0)>0$, a simple calculus exercise shows that
$$ h_{\text{opt,TBD}}(N)=  \Big( \frac{C_2 }{4 C_1} \Big)^{1/5}  \Big( \frac{\sigma_+^2(0) + \sigma_-^2(0)}{f(0) \big( \mu_+^{(2)}(0)- \mu_-^{(2)}(0)\big)^2} \Big)^{1/5}  N^{-1/5}.  \qed $$

\section{Proof of Theorem~\ref{theorem:MSE_RDD_versu_TBD}}\label{sec:relMSETheoremProof}

First note that by plugging \eqref{eq:hopt_RDD_est} into \eqref{eq:AMSE_RDD_def}, $\text{AMSE}_{\text{RDD}} \big( \hat{h}_{\text{opt,RDD}}(N),N \big)$ equals

$$\begin{aligned} & \tilde{C}_1 \big( \mu_+^{(2)}(t) -\mu_-^{(2)}(t) \big)^2 \hat{h}^4_{\text{opt,RDD}}(N)+ \frac{\tilde{C}_2}{N \hat{h}_{\text{opt,RDD}}(N)} \Big( \frac{\sigma_+^2(t)+\sigma_-^2(t)}{f(t)} \Big) \\ 
& =  \big( 4 \tilde{C}_1 \tilde{C}_2^4 \big)^{1/5} N^{-4/5} \Big[ \frac{1}{4} \big( \mu_+^{(2)}(t) -\mu_-^{(2)}(t) \big)^2 \hat{\gamma}_{\text{RDD}}^4 + \frac{\sigma_+^2(t)+\sigma_-^2(t)}{\hat{\gamma}_{\text{RDD}} f(t)}  \Big] \\ 
& = \big( 4 \tilde{C}_1 \tilde{C}_2^4 \big)^{1/5} N^{-4/5} \Big[ \frac{1}{4} \big( \mu_+^{(2)}(t) -\mu_-^{(2)}(t) \big)^2 \big( \gamma+o_p(1) \big)^4 + \frac{\sigma_+^2(t)+\sigma_-^2(t)}{\big( \gamma+o_p(1) \big) f(t)}  \Big] \\ 
& = \big( 4 \tilde{C}_1 \tilde{C}_2^4 \big)^{1/5} N^{-4/5} \Big[ \frac{1}{4} \big( \mu_+^{(2)}(t) -\mu_-^{(2)}(t) \big)^2 \gamma^4 + \frac{\sigma_+^2(t)+\sigma_-^2(t)}{ \gamma  f(t)}  +o_p(1)\Big] \\ 
& =  \big( 4 \tilde{C}_1 \tilde{C}_2^4 \big)^{1/5} N^{-4/5} \big[\alpha +o_p(1) \big], \end{aligned}$$ 
where the third equality holds since $\hat{\gamma}_{\text{RDD}} \xrightarrow{p} \gamma$ and the last inequality uses definitions \eqref{eq:alpha_definition} and \eqref{eq:def_gamma}.

Plugging \eqref{eq:hopt_TBD_est} into \eqref{eq:AMSE_TBD_def}, an identical argument that uses $\hat{\gamma}_{\text{TBD}} \xrightarrow{p} \gamma$ and definitions \eqref{eq:alpha_definition} and \eqref{eq:def_gamma} yields $$\text{AMSE}_{\text{TBD}} \big( \hat{h}_{\text{opt,TBD}}(\theta N), \theta N \big) = \big( 4 C_1 C_2^4 \big)^{1/5} (\theta N)^{-4/5} \big[ \alpha +o_p(1) \big].$$

In applying Lemma \ref{lemma:MSE_RDD_IK}, it will be helpful to note that letting $c=\big( \tilde{C}_2/(4 \tilde{C}_1) \big)^{1/5}$, $$\begin{aligned}  \hat{h}_{\text{opt,RDD}}^4(N) + \frac{1}{N \hat{h}_{\text{opt,RDD}}(N) }  & =  c^4   \hat{\gamma}_{\text{RDD}}^4 N^{-4/5} +  \frac{1}{c \hat{\gamma}_{\text{RDD}}} N^{-4/5}  
\\ & =  \Big( c^4 \big(\gamma+o_p(1) \big)^4 + \frac{1}{c \big(\gamma+o_p(1) \big)}  \Big) N^{-4/5}
\\ & = \big( c^4 \gamma +\frac{1}{c \gamma} +o_p(1) \big) N^{-4/5} \\ & = O_p(N^{-4/5}). \end{aligned}$$ 
Therefore under the conditions of Theorem \ref{theorem:MSE_RDD_versu_TBD}, by applying Lemmas \ref{lemma:MSE_RDD_IK} and \ref{lemma:MSE_TBD} and combining the previous results, $$ \begin{aligned} \frac{\text{MSE}_{\text{RDD}} \big( \hat{h}_{\text{opt,RDD}}(N),N \big)}{\text{MSE}_{\text{TBD}} \big( \hat{h}_{\text{opt,TBD}}(\theta N),\theta N \big)} & = \frac{\text{AMSE}_{\text{RDD}} \big( \hat{h}_{\text{opt,RDD}}(N),N \big) + o_p( O_p(N^{-4/5}) ) }{\text{AMSE}_{\text{TBD}} \big( \hat{h}_{\text{opt,TBD}}(\theta N),\theta N \big)+o_p ( O_p(N^{-4/5}) )} 
\\ & = \frac{ \big( 4 \tilde{C}_1 \tilde{C}_2^4 \big)^{1/5} N^{-4/5} \big[\alpha +o_p(1) \big] + N^{-4/5}o_p(1)}{\big( 4 C_1 C_2^4 \big)^{1/5} (\theta N)^{-4/5} \big[\alpha +o_p(1) \big]+N^{-4/5}o_p(1)} 
\\ & = \frac{ \big( 4 \tilde{C}_1 \tilde{C}_2^4 \big)^{1/5}  \big[\alpha +o_p(1) \big] + o_p(1)}{\big( 4 C_1 C_2^4 \big)^{1/5} \theta ^{-4/5} \big[\alpha  +o_p(1) \big]+o_p(1)} 
\\ & = \theta^{4/5} \Big( \frac{\tilde{C}_1 \tilde{C}_2^4}{C_1 C_2^4} \Big)^{1/5} +o_p(1). \qed \end{aligned}$$

\section{Extensions to assignment probability $p \neq 1/2$}\label{sec:GenPneqHalf}
In this appendix we consider the MSE of a generalization of the 3-level TBD given by \eqref{eq:threelevel}, where we allow the assignment probabilities to $Z=1$ and $Z=-1$ to differ from $1/2$ within the interval of experimentation. In particular, we consider a design with the following assignment probabilities \begin{align}\label{eq:pzgivx_KO_anyP}
\Pr( Z_i=1\giv x_i)=\begin{cases}
0, & x_i -t\le -\Delta\\
p, & |x_i-t|<\Delta\\
1, & x_i-t \ge \Delta,
\end{cases}
\end{align} for some $p \in (0,1)$. Under this more general version of the TBD, one can show that the AMSE formula given in \eqref{eq:AMSE_TBD_def} changes to \begin{equation}\label{eq:AMSE_TBD_defGenP}
    \text{AMSE}_{\text{TBD}(p)}(h,N) \equiv C_1 \big(  \mu_+^{(2)}(t) -\mu_-^{(2)}(t) \big)^2 h^4+ \frac{C_2}{Nh} \Big( \frac{ \frac{1}{2p} \sigma_+^2(t)+\frac{1}{2(1-p)}\sigma_-^2(t)}{f(t)} \Big),
\end{equation} and that the following lemma holds.

 \begin{lemma}\label{lemma:MSE_TBD_genP} 
 Under the conditions of Lemma \ref{lemma:MSE_TBD}, with the exception that the assignment probabilities follow \eqref{eq:pzgivx_KO_anyP} rather than \eqref{eq:threelevel}, the mean squared error in estimating $\tau_{\emph{thresh}}$ is given by 
 \begin{equation*}
     \emph{MSE}_{\emph{TBD}(p)}(h,N)= \emph{AMSE}_{\emph{TBD}(p)}(h,N) + o_p\Big( h^4 + \frac{1}{Nh} \Big),
 \end{equation*} and the asymptotically optimal bandwidth, defined by $\argmin_h \emph{AMSE}_{\emph{TBD}(p)}(h,N)$ is \begin{equation}\label{eq:hopt_TBD_genP}  h_{\emph{opt,TBD}(p)}(N)=  \Big( \frac{C_2}{4 C_1} \Big)^{1/5}  \Big( \frac{ \frac{1}{2p} \sigma_+^2(t) +\frac{1}{2(1-p)}\sigma_-^2(t)}{f(t) \big( \mu_+^{(2)}(t)- \mu_-^{(2)}(t)\big)^2} \Big)^{1/5}  N^{-1/5}.  \end{equation}
 \end{lemma}

 \begin{proof}
 The proof is identical to the proof of Lemma \ref{lemma:MSE_TBD} in Appendix \ref{sec:proof_of_TBD_MSE}, except that the formulas for the functions $f_{\pm}(\cdot)$ presented at \eqref{eq:fpm_def} are instead given by 
 \begin{equation*}
    f_+(x) = \begin{cases} 
\frac{f(x)}{\Pr(Z_i=1)},  &  x \geq \Delta \\ \frac{pf(x)}{ \Pr(Z_i=1)},  &  |x|<\Delta \\  0, &  x \leq - \Delta \end{cases} 
\quad  \text{and} \quad  
f_-(x) = \begin{cases} 0,  &  x \geq \Delta \\ 
\frac{(1-p)f(x)}{ \Pr(Z_i=-1)},  &  |x|<\Delta \\  \frac{f(x)}{\Pr(Z_i=-1)}, &  x \leq - \Delta. \end{cases}  
\end{equation*} The changes in the formulas for $f_{\pm}(\cdot)$ do not affect the leading order bias formula but do affect the leading order variance and optimal bandwidth formulas.
 \end{proof}

Analogously to the empirical bandwidth choice given by \eqref{eq:hopt_TBD_est}, an investigator running a TBD of the form \eqref{eq:pzgivx_KO_anyP} seeking mean squared optimal estimation of $\tau_{\thresh}$ would ultimately use the bandwidth $$\hat{h}_{\text{opt,TBD}(p)}(N)= \Big(\frac{C_2}{4C_1} \Big)^{1/5} \hat{\gamma}_{\text{TBD}(p)} N^{-1/5},$$ where $\hat{\gamma}_{\text{TBD}(p)}$ is some consistent estimator of the quantity $$\Big( \frac{ \frac{1}{2p} \sigma_+^2(t) +\frac{1}{2(1-p)}\sigma_-^2(t)}{f(t) \big( \mu_+^{(2)}(t)- \mu_-^{(2)}(t)\big)^2} \Big)^{1/5}.$$ 

To derive an analogue of Theorem \ref{theorem:MSE_RDD_versu_TBD} for 
a TBD of the form \eqref{eq:pzgivx_KO_anyP}, it is 
convenient to define the relative variance measure 
\begin{equation}\label{eq:varRatDef}
r(t) \equiv \frac{\sigma_+^2(t)}{\sigma_+^2(t)+\sigma_-^2(t)}.    
\end{equation} The following theorem compares the RDD with $N$ points to a TBD of the form \eqref{eq:pzgivx_KO_anyP} with $\theta N$ points for some $\theta>0$.
\begin{theorem}\label{theorem:MSE_RDD_versu_TBD_GenP}
Under the assumptions of Theorem \ref{theorem:MSE_RDD_versu_TBD}, as $N \to \infty$
\begin{equation}\label{eq:MSE_rat_GenP}
    \frac{\emph{MSE}_{\emph{RDD}}\big(\hat{h}_{\emph{opt,RDD}}(N),N \big)}{\emph{MSE}_{\emph{TBD}(p)}\big(\hat{h}_{\emph{opt,TBD}(p)}(\theta N), \theta N \big)} \xrightarrow{p} \theta^{4/5} \Big( \frac{\tilde{C}_1 \tilde{C}_2^4}{C_1 C_2^4}\Big)^{1/5} \Big( \frac{r(t)}{2p} +\frac{1-r(t)}{2(1-p)}\Big)^{-4/5}
\end{equation}
holds for any tie-breaker design
of the form \eqref{eq:pzgivx_KO_anyP} with $\Delta>0$. 
\end{theorem}

\begin{proof}
The proof is the same as that of Theorem \ref{theorem:MSE_RDD_versu_TBD}, except the formulas for $\text{AMSE}_{\text{TBD}(p)}(h,N)$, $h_{\text{opt,TBD}(p)}(N)$, and $\hat{h}_{\text{opt,TBD}(p)}(N)$ are different in the setting of a tie-breaker design of the form $\eqref{eq:pzgivx_KO_anyP}$.
\end{proof}

We visualize the implications of Theorem \ref{theorem:MSE_RDD_versu_TBD_GenP} in Figure \ref{fig:pNeqHalf}. While the figure focuses on the case of the triangular kernel with $\theta=1$, the relative AMSE as a function of $p$ and $r(t)$ will be a scalar multiple of that shown in Figure \ref{fig:pNeqHalf} for different kernels and values of $\theta$. We also explicitly list a few notable implications of Theorem \ref{theorem:MSE_RDD_versu_TBD_GenP} below. 

\begin{figure}[t]
\centering
\includegraphics[width=0.9 \hsize]{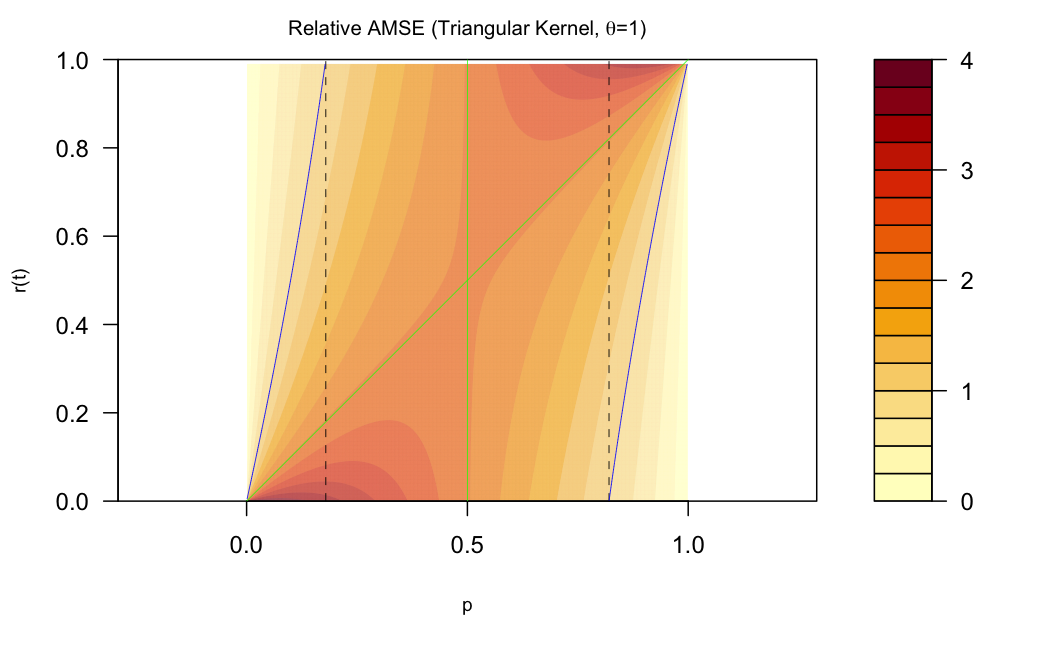}
\caption{\label{fig:pNeqHalf} Relative AMSE given by Theorem \ref{theorem:MSE_RDD_versu_TBD_GenP} as function of $p$ and $r(t)$. Here we consider the same sample size for each design ($\theta=1$), and that the most favorable kernel for the RDD (the triangular kernel) is to be used. The color scale for the contour plot is displayed on the right, with the added blue lines denoting a level of $1$. Outside the blue lines, the RDD performs better than the TBD asymptotically. The vertical black lines, between which the relative AMSE is always greater than 1, are placed at $p=\theta_*/2$ and $p=1-\theta_*/2$, for $\theta_*=60.46618^{-1/4}$. The pair of green lines give the boundaries of the two triangular regions in which the TBD with assignment probabilities \eqref{eq:pzgivx_KO_anyP} has a smaller AMSE than the default TBD with $p=1/2$ has.
}
\end{figure}

\begin{enumerate}
    \item The AMSE for the TBD is minimized by $p=\sigma_+(t)/\big(\sigma_-(t)+\sigma_+(t) \big)$.
    
    \item If either $r(t) \leq p \leq 1/2$ or $1/2 \leq p \leq r(t)$, then  $$\Big( \frac{r(t)}{2p} +\frac{1-r(t)}{2(1-p)}\Big)^{-4/5} \geq 1.$$ As a result, whenever $p$ is between $r(t)$ and $1/2$ inclusive, the advantage of a TBD over an RDD (in terms of relative asymptotic MSE in estimation of $\tau_{\thresh}$) is at least as large as that described in Table \ref{table:Kernels}, which corresponds to the $p=1/2$ case.
    
    \item Suppose that a TBD and RDD are to be conducted with the same sample size $N$ (i.e.\ $\theta=1$), and that under both designs the asymptotically MSE optimal bandwidth is to be used. Then, regardless of the unknown value of $r(t) \in [0,1]$, for any choice of $p \in (\theta_*/2,1-\theta_*/2)$ with $\theta_*$ defined at \eqref{eq:DefthetaStar}, the TBD given by \eqref{eq:pzgivx_KO_anyP} will have lower AMSE than the RDD. Conversely, if $p \notin [\theta_*/2,1-\theta_*/2]$, there will exist a range of values of $r(t)$ contained in $[0,1]$ for which the TBD will have higher AMSE than the RDD. Note that amongst the common kernel choices considered in Table \ref{table:Kernels}, $\theta_*$ is largest for the triangular kernel with a value of precisely $60.46618^{-1/4}$, suggesting that for any of these kernel choices the TBD will have lower AMSE than an RDD in estimating $\tau_{\thresh}$ as long as $p \in [0.18,0.82]$, but outside that range, the TBD may have larger AMSE depending on the kernel choice and the unknown value of $r(t)$.
\end{enumerate}

The first fact may be difficult to leverage in practice because at the time an experiment is being designed, the investigator is unlikely to have outcome data from both treatment and control units that they can use to estimate the ratio $\sigma_+(t)/\sigma_-(t)$. If the investigator has some prior knowledge about whether or not $\sigma_+(t)/\sigma_-(t)$ is smaller or greater than 1, they could potentially leverage the second fact to design a TBD that has lower AMSE than a TBD with the default of $p=1/2$ has. The third statement is true for any value of $r(t)$ or $\sigma_+(t)/\sigma_-(t)$, suggesting that a design with $p \notin [0.18,0.82]$ should only be considered with caution.

\section{Bandwidth that removes leading order bias}\label{sec:MagicBandwidth}

In the asymptotic analysis from Section \ref{sec:bandwidth_shrink_asymptotics}, the bias for the local linear estimator of $\hat{\tau}_{\thresh}$ was $O_p(h^2)$. However, that analysis relied on a simplification that arose by assuming that $h \to 0$ as $N \to \infty$. In this appendix, we derive an asymptotic expression for the bias of $\hat{\tau}_{\thresh}$ that does not rely on the assumption that $h \to 0$ as $N \to \infty$. In this expression for the bias, the leading order term is quadratic in $h$. We also show that for some value of $h/\Delta$, which depends only on the kernel $K(\cdot)$, the quadratic bias term vanishes. In Table \ref{table:MagicBandwidth}, we present a numerical approximation of the value of $h/\Delta$ that removes the leading order bias for common kernel choices.

 We consider the following regularity conditions, which were not needed in the main text.
\begin{enumerate}[(I)]

    \item The mean functions $\mu_{\pm}(\cdot)$ are three times differentiable. Further, the third derivatives of the mean functions, $\mu_{\pm}^{(3)}(\cdot)$, are bounded.
    \item The density $f$ of the assignment variable satisfies $\vert f(x) -f(t) \vert \leq L \vert x -t \vert$ for some $L<\infty$.
\end{enumerate}

In the proof of Lemma \ref{lemma:MSE_TBD}, where we computed the leading order bias, we did not need to assume (I) and (II) because we supposed that $h \to 0$ as $N \to \infty$. In the following theorem about the leading order bias, because we do not assume $h \to 0$, we assume (I) and (II) in addition to (i)--(vi).

\begin{theorem}\label{theorem:magicBandwidth}
    Suppose conditions (i)--(vi) from the main text and conditions (I) and (II) hold, and that $N \to \infty$. Let $\Delta>0$ be fixed. Under the tie-breaker design defined by \eqref{eq:threelevel} and estimation of $\tau_{\thresh}$ according to \eqref{eq:kernelcriterion} and \eqref{eq:def_hat_tauthresh} for some bandwidth $h>0$, the bias of $\hat{\tau}_{\thresh}$ is given by \begin{equation}\label{eq:2ndOrderBias_withDelta}   \frac{ \mu_+^{(2)} (t)-\mu_-^{(2)} (t)}{2}   \Big( \frac{ \nu_2^2 - 4 \int_{\Delta/h}^{\infty} uK(u) \rd u  \int_{\Delta/h}^{\infty} u^3K(u) \rd u }{ \nu_0 \nu_2 -4[ \int_{\Delta/h}^{\infty} uK(u) \rd u ]^2  } \Big)h^2 +O(h^3)+O_p\Big( \frac{1}{\sqrt{N}} \Big), \end{equation} where $\nu_j$ is defined at \eqref{eq:def_nuAndPi_newer}. As a result, the leading order bias term is equal to $0$ whenever the bandwidth $h$ is chosen so that \begin{equation}\label{eq:cond_no2ndOrderBias}
        \int_{\Delta/h}^{\infty} uK(u) \rd u  \int_{\Delta/h}^{\infty} u^3K(u) \rd u= \bigg( \int_0^{\infty} u^2 K(u) \rd u \bigg)^2.
    \end{equation} Moreover, there must exist an $h> \Delta >0$ that satisfies \eqref{eq:cond_no2ndOrderBias}.
\end{theorem}

\begin{proof}

Fix $\Delta >0$ and suppose without loss of generality that $t=0$. Define $K_h(\cdot)$, $\mathcal{Z}_{\pm}$, $N_{\pm}$, $F_{j,\pm}$ for $j=0,1,\dots,4$, $f_{\pm}(\cdot)$, $\cx_{\pm}$,$W_{\pm}$, $\cy_{\pm}$, $\hat{\mu}_{\pm}(0)$, and $B_{\pm}$, according to the same definitions presented in Appendix \ref{sec:proof_of_TBD_MSE}. 

We will first derive a formula for $F_{j,\pm}$ that is similar to that given in Lemma \ref{lemma:F_formula}, except here we will not rely on a simplification that occurs in the asymptotic regime where $h \to 0$, eventually dropping below $\Delta$.

For any nonnegative integer $j$, because the samples are IID, we get
$$\begin{aligned} F_{j,+} & =  \frac{1}{N_+} \sum_{i \in \mathcal{Z}_+} X_i^j K_h(X_i) 
\\ & =\e[X^j K_h(X)|Z=1]+O_p\Big( \sqrt{\text{Var}[X^j K_h(X)|Z=1]/N_+} \Big)
\\ & = \frac{1}{h} \int_{-\infty}^{\infty} x^j K(x/h) f_+(h) \rd x + O_p \big( 1/\sqrt{N_+} \big)
\\ & = h^j \int_{-\infty}^{\infty} u^j K(u) f_+(hu) \rd u + O_p(1/\sqrt{N})
\\ & = \frac{h^j \big( \int_{-\Delta/h}^{\Delta/h} u^j K(u) f(hu) \rd u +2 \int_{\Delta/h}^{\infty} u^j K(u) f(hu) \rd u \big)}{2 \Pr(Z_i=1)} +O_p(1/\sqrt{N}). \end{aligned}$$

To simplify this expression further, observe that, for any $a,b \in [-\infty,\infty]$ and nonnegative integer $j$, $$\begin{aligned}
  \bigg|  \int_{a}^{b} u^j K(u) \big[ f(hu)-f(0) \big] \rd u \bigg| & = \bigg| h \int_{a}^{b} u^{j+1} K(u) \Big[ \frac{f(hu)-f(0)}{hu} \Big] \rd u \bigg|
   \\ & \leq h \int_{a}^{b}  \Big| \frac{f(hu)-f(0)}{hu} \Big| \cdot \vert u^{j+1} \vert K(u) \rd u
    \\ & \leq L h \int_{a}^{b}   \vert u^{j+1} \vert K(u) \rd u
    \\ & = O(h),
\end{aligned}$$
where above the last two steps hold by Assumptions (II) and (iv) respectively, and as a consequence,
for any $a<b$,
$$\begin{aligned} \int_{a}^{b} u^j K(u)  f(hu) \rd u & = \int_{a}^{b} u^j K(u)  f(0) \rd u+ \int_{a}^{b} u^j K(u) \big[ f(hu)-f(0) \big] \rd u 
\\ & = f(0) \int_{a}^{b} u^j K(u)   \rd u +O(h).\end{aligned}$$ Combining this with the above formula for $F_{j,+}$,

$$ \begin{aligned} F_{j,+} & = \frac{h^j f(0) \big( \int_{-\Delta/h}^{\Delta/h} u^j K(u) \rd u +2 \int_{\Delta/h}^{\infty} u^j K(u)  \rd u +O(h) \big)}{2 \Pr(Z_i=1)} +O_p(1/\sqrt{N}) \\ & = h^j f_+(0) \nu_{j,+}(\Delta/h)+O(h^{j+1})+ O_p(1/\sqrt{N}), \end{aligned}$$ where $$\nu_{j,+}(\delta) \equiv \int_{-\delta}^{\delta} u^j K(u) \rd u +2 \int_{\delta}^{\infty} u^j K(u) \rd u.$$

If we define $\nu_{j,-}(\delta) \equiv \int_{-\delta}^{\delta} u^j K(u) \rd u +2 \int_{-\infty}^{-\delta} u^j K(u) \rd u,$ a similar argument yields an analogous formula for $F_{j,-}$, and hence \begin{equation}\label{eq:Fj_pm_wDelta}
    F_{j,\pm} = h^j f_{\pm}(0) \nu_{j,\pm}(\Delta/h)+O(h^{j+1})+O_p(1/\sqrt{N}) \quad \text{for all }  j \in \mathbb{Z}_{\geq 0}.
\end{equation}

By noting the similarities between the formula for $F_{j,+}$ in Lemma \ref{lemma:F_formula} and that from \eqref{eq:Fj_pm_wDelta}, the formula for $B_+$ can be derived by the same argument presented in Appendix \ref{sec:proof_of_TBD_MSE}, by simply replacing $\nu_j$ terms with $\nu_{j,+}(\Delta/h)$, $o_p(h^j)$ terms with $O(h^{j+1})+O_p(1/\sqrt{N})$ terms, and intermediary $O_p(h^j)$ terms with $O(h^j)+O_p(1/\sqrt{N})$, for any integer $j$. Because we now no longer assume $h \to 0$ as $N \to \infty$, there are three steps where we have to use a slightly different argument from that presented in Appendix \ref{sec:proof_of_TBD_MSE}. First, since we no longer assume that $h \to 0$, the 2nd order Taylor expansion of the mean functions with the remainder terms $T_i$ require that $\mu_{\pm}^{(2)}(\cdot)$ is continuous everywhere, rather than merely in a neighborhood of $t$ which is given by Assumption (iii). Assumption (I) that $\mu_{\pm}^{(3)}(\cdot)$ is bounded, guarantees that $\mu_{\pm}^{(2)}(\cdot)$ is continuous everywhere because differentiability implies continuity. Second, we can use Assumption (I) that $\mu_{\pm}^{(3)}(\cdot)$ is bounded to show that there exists an $\bar{a}_+ < \infty$ for which $K_h(X_i) \sup_{x \in [-X_i,X_i]} \vert \mu_{+}^{(3)}(x) \vert  \leq \bar{a}_+ K_h(X_i)$ for any $h >0$. In  Appendix \ref{sec:proof_of_TBD_MSE}, we did not need Assumption (I) since we supposed that $h \to 0$, so it was enough to show that inequality holds for any sufficiently small $h$. Third, the argument involves a first order Taylor expansion of the quantity $(F_{0,+} F_{2,+} -F_{1,+}^2)^{-1}$. In the current setting where we do not suppose $h \to 0$ as $N \to \infty$, we must now ascertain that with probability approaching $1$ as $N \to \infty$, $F_{0,+} F_{2,+} -F_{1,+}^2$ is positive and bounded away from zero. Recalling that up to an $O_p(1/\sqrt{N})$ term, $F_{j,+}$, equals $h^j \int_{-\infty}^{\infty} u^j K(u) f_+(hu) \rd u$, by the Cauchy-Schwartz inequality it follows that $$\begin{aligned} F_{1,+}^2  & = h^2 \bigg( \int_{-\infty}^{\infty} u K(u) f_+(hu) \rd u \bigg)^2  + O_p(1/\sqrt{N})
\\ & = h^2 \bigg( \int_{-\infty}^{\infty} \sqrt{K(u) f_+(hu) } \cdot \sqrt{u^2 K(u) f_+(hu)} \rd u \bigg)^2 + O_p(1/\sqrt{N})
\\ & <  h^2 \int_{-\infty}^{\infty} K(u) f_+(hu) \rd u  \int_{-\infty}^{\infty} u^2 K(u) f_+(hu) \rd u + O_p(1/\sqrt{N})
\\ & = F_{0,+} F_{2,+}+O_p(1/\sqrt{N}). \end{aligned}$$  
Hence with probability converging to $1$ as $N \to \infty$, $F_{0,+} F_{2,+}-F_{1,+}^2$ is positive and bounded away from zero. Consequently, the following first order Taylor expansion holds for all $h>0$ $$\begin{aligned} \frac{h^2 f_+^2(0)}{F_{0,+} F_{2,+}-F_{1,+}^2} & = \frac{1}{[ \nu_{0,+}(\Delta/h) \nu_{2,+}(\Delta/h)- \nu_{1,+}(\Delta/h)^2] +O(h)+O_p(1/\sqrt{N})} \\ & = \frac{1}{ \nu_{0,+}(\Delta/h) \nu_{2,+}(\Delta/h)- \nu_{1,+}(\Delta/h)^2 } +O(h)+O_p(1/\sqrt{N}), \end{aligned}$$ where the denominator in the first term is positive and bounded away from $0$ by a similar Cauchy-Schwartz argument that leverages symmetry of $K$.

By the same derivation of the bias formula presented in Appendix \ref{sec:proof_of_TBD_MSE}, with the exceptions in the argument noted above, the bias of $\hat{\mu}_+(0)$ is given by $$B_+ = \frac{1}{2} \mu_+^{(2)}(0) \Big( \frac{[\nu_{2,+}(\Delta/h)]^2 - \nu_{1,+}(\Delta/h) \nu_{3,+}(\Delta/h)}{\nu_{0,+}(\Delta/h) \nu_{2,+}(\Delta/h) - [\nu_{1,+}(\Delta/h)]^2} \Big)h^2 + O(h^3)+O_p(1/\sqrt{N}).$$ 

A similar argument shows that, $$B_- = \frac{1}{2} \mu_-^{(2)}(0) \Big( \frac{[\nu_{2,-}(\Delta/h)]^2 - \nu_{1,-}(\Delta/h) \nu_{3,-}(\Delta/h)}{\nu_{0,-}(\Delta/h) \nu_{2,-}(\Delta/h) - [\nu_{1,-}(\Delta/h)]^2} \Big)h^2 + O(h^3)+O_p(1/\sqrt{N}).$$

Now note that the bias of $\hat{\tau}_{\thresh}$ is given by $B_+-B_-$. To simplify the expression for $B_+-B_-$, observe that by symmetry of $K(\cdot)$, for any $h>0$ $$\nu_{j,-}(\Delta/h)=\begin{cases} \nu_{j,+}(\Delta/h)=\nu_j  & \text{ for even } j, \\ -\nu_{j,+}(\Delta/h) = -2 \int_{\Delta/h}^{\infty} u^j K(u) \rd u  & \text{ for odd } j. \end{cases}$$ 
The $\nu_j$ above is defined at \eqref{eq:def_nuAndPi_newer}. Hence, subtracting $B_-$ from $B_+$ the bias for $\hat{\tau}_{\thresh}$ is given by $$\frac{h^2}{2} \big( \mu_+^{(2)} (0)-\mu_-^{(2)} (0) \big)  \frac{ \nu_2^2 - 4 \int_{\Delta/h}^{\infty} uK(u) \rd u  \int_{\Delta/h}^{\infty} u^3K(u) \rd u }{ \nu_0 \nu_2 -4[ \int_{\Delta/h}^{\infty} uK(u) \rd u ]^2  } +O(h^3)+O_p(1/\sqrt{N}).$$ Since we supposed $t=0$ without loss of generality, this proves \eqref{eq:2ndOrderBias_withDelta}.

 Now, since $\nu_0 \nu_2 -4[ \int_{\Delta/h}^{\infty} uK(u) \rd u ]^2$ is bounded away from zero, the leading order term in $h$ in formula \eqref{eq:2ndOrderBias_withDelta} for the bias of $\hat{\tau}_{\thresh}$ is zero whenever $h$ satisfies \eqref{eq:cond_no2ndOrderBias}. 

To prove the final claim, that there exists an $h > \Delta$ solving \eqref{eq:cond_no2ndOrderBias}, define $$ \varrho(h) \equiv 4 \int_{\Delta/h}^{\infty} uK(u) \rd u  \int_{\Delta/h}^{\infty} u^3K(u) \rd u \quad \text{ for } h>0.$$ 
Clearly by boundedness and continuity of $K(\cdot)$, $\varrho(\cdot)$ is continuous on $[0,\infty)$. Recalling Assumption (iv) that $K(\cdot)$ is supported on $[-1,1]$, $\varrho(\Delta)=0$.  Since $\nu_2>0$, $\varrho(\Delta) < \nu_2^2$. Also note that by symmetry of the kernel and the Cauchy-Schwartz inequality, $$\begin{aligned} \nu_2^2 & = \bigg( 2 \int_0^{\infty} u^2 K(u) \rd u \bigg)^2
\\ & = 4 \bigg( \int_0^{\infty} \sqrt{u K(u)} \cdot \sqrt{u^3 K(u)} \rd u \bigg)^2
\\ & < 4  \int_0^{\infty} u K(u) \rd u  \int_0^{\infty} u^3 K(u) \rd u 
\\ & =  \lim\limits_{h \to \infty} \varrho(h).
\end{aligned}$$
Since $\varrho(\Delta)< \nu_2^2 < \lim_{h \to \infty} \varrho(h)$, and since $\varrho(\cdot)$ is continuous on $[\Delta,\infty)$, 
there must exist some $\tilde{h} \in (\Delta,\infty)$ with $\varrho(\tilde{h})=\nu_2^2$. Clearly this $\tilde{h}$ satisfies equation \eqref{eq:cond_no2ndOrderBias} and $\tilde{h}>\Delta$.
\end{proof}

\begin{table}[!t]
\caption{Table of bandwidths that remove the leading order bias term, rounded to 4 decimal points in units of $\Delta$, for various kernel choices.}
\label{table:MagicBandwidth}
\begin{center}
\begin{tabular}{ l l c c } 
\toprule
 Kernel & Function $K(u)$ & Support & Bandwidth solving \eqref{eq:cond_no2ndOrderBias} \\ 
\midrule
 Boxcar & $1/2$ & $[-1,1]$  & $3.1342 \Delta$ \\ [1ex] 

Triangular & $(1-\vert u \vert)_+$ & $[-1,1]$  & $4.0343 \Delta$ \\  [1ex] 

 Epanechnikov & $ \frac{3}{4}(1- u^2)_+$ & $[-1,1]$ & $3.7866 \Delta$ \\ [1ex]

 Quartic & $ \frac{15}{16} (1-u^2)_+^2$  & $[-1,1]$ &  $4.3432 \Delta$ \\ [1ex]

 Triweight & $ \frac{35}{32} (1-u^2)_+^3$  & $[-1,1]$ & $4.8367 \Delta$ \\ [1ex]

  Tricube & $ \frac{70}{81} (1-\vert u\vert^3 )_+^3$ & $[-1,1]$ &  $4.3805 \Delta$ \\ [1ex]

 Cosine & $\frac{\pi}{4} \cos \big( \frac{\pi}{2} u \big)$  & $[-1,1]$ & $3.8613 \Delta$  \\ 
 \bottomrule
 
\end{tabular}
\end{center}
\end{table}

\section{Proof of Proposition~\ref{prop:triangleefficiency}}\label{sec:TSefficiencyCalc}

Note that for the triangular kernel, $\lambda_0=(2/3)\kappa_0$ and $\lambda_2=(2/5)\kappa_2$. By substituting these expressions into formula~\eqref{eq:general_efficiency_ratio} for the efficiency ratio and cancelling out a common factor of $2\kappa_2/15$ from the numerator and denominator, it follows that

\begin{align}\label{eq:triefffactors}
 \eff_{\ts} &=\frac{ 5 \kappa_0 \kappa_2-15 \phi(0) \psi(0) +3 \phi^2(0) }{ 5 \kappa_0 \kappa_2 -15  \phi(\Delta) \psi(\Delta) +3 \phi^2(\Delta) }
\times\frac{ (\kappa_0 \kappa_2 - \phi^2(\Delta)  )^{2} }{   (\kappa_0 \kappa_2 - \phi^2(0) )^{2}}.
\end{align}
Next $5\kappa_0\kappa_2-15\phi(\Delta)\psi(\Delta)+3\phi^2(\Delta)$ equals
\begin{align*}
    \frac{5h^4}{24}-15\frac{h^4}{72}(1-3\delta^2+2\delta^3)(1-6\delta^2+8\delta^3-3\delta^4)
    +\frac{h^4}{12}(1-3\delta^2+2\delta^3)^2,
\end{align*}
and so the first factor in~\eqref{eq:triefffactors} is
$$
\frac{ 2  }
{ 5-5(1-3\delta^2+2\delta^3)(1-6\delta^2+8\delta^3-3\delta^4)+2(1-3\delta^2+2\delta^3)^2 }.
$$
Turning to the second factor
\begin{align*}
    \kappa_0\kappa_2-\phi^2(\Delta)&=
    \frac{h^4}{24}-\frac{h^4}{36}(1-3\delta^2+2\delta^3)^2
    =\frac{h^4}{72}\bigl(3-2(1-3\delta^2+2\delta^3)^2\bigr),
\end{align*}
and so the second factor equals
$
(3-2(1-3\delta^2+2\delta^3)^2)^2
$, establishing~\eqref{eq:triangleefficiency}.

\section{Proof of Proposition~\ref{prop:itsmonotone}}\label{sec:MonotoneEffTS}

We want to show that this function
$$\eff_{\ts}(\delta) = \frac{2 \big(3-2 (1- 3 \delta^2 +2 \delta^3)^2 \big)^2}{5- 5  (1- 3 \delta^2 +2 \delta^3)
(1 - 6 \delta^2+ 8 \delta^3 -3 \delta^4 )
+2 (1- 3 \delta^2 +2 \delta^3)^2} $$
has a positive derivative for $0<\delta<1$.
The numerator has degree $12$ and the denominator has degree $7$.
The customary formula for the derivative of a rational function
produces a rational function with a non-negative denominator and
a numerator of degree $18$.  We will work through a sequence
of steps reducing the degree of this polynomial
to show that the numerator must be positive on $(0,1)$.
That then rigorously establishes the monotonicity of $\eff_{\ts}(\delta)$
which is visually apparent.

It is convenient to work instead with $x=1-\delta$. Then
$1-3\delta^2+2\delta^3
=3x^2-2x^3
$
and
$1 - 6 \delta^2+ 8 \delta^3 -3 \delta^4
=4x^3-3x^4.$
Therefore $\eff_{\ts}(\delta)=f(1-\delta)$ where $f$ is a function given by
\begin{align*}
f(x)&=\frac{2(3-2(3x^2-2x^3)^2)^2}
{
5-5(3x^2-2x^3)(4x^3-3x^4)+2(3x^2-2x^3)^2
}\\
&=\frac{2(3-2x^4(3-2x)^2)^2}{
5-5x^5(3-2x)(4-3x)+2x^4(3-2x)^2}\\
&=\frac{2(3-2x^4(3-2x)^2)^2}{
5+[x^4(3-2x)][6 -24x+15x^2]}\\
&=\frac{2(3-2g(x) (3-2x))^2}{
5+g(x)(6 -24x+15x^2)}
\end{align*}
for $g(x) = x^4(3-2x)$
and having replaced $\delta$ by $x=1-\delta$ we will show that $f'(x)<0$.

In the usual formula for the derivative of a ratio, $f'(x)$ has
this numerator
\begin{align*}
n_1(x) &=
4(3-2g(x)(3-2x))(4g(x)-2g'(x)(3-2x))[5+g(x)(6 -24x+15x^2)]\\
&-2(3-2g(x)(3-2x))^2[g'(x)(6-24x+15x^2)+g(x)(-24+30x)].
\end{align*}
Notice that $0\le g(x)(3-2x) \le1$ for $x \in [0,1]$ and
so $3-2g(x)(3-2x)>0$.
As a result, the sign of $n_1(x)$ is preserved
by dividing it by $2(3-2g(x)(3-2x))$, yielding
\begin{align*}
n_2(x)
  &= 2(4g(x)-2g'(x)(3-2x))[5+g(x)(6 -24x+15x^2)]\\
 &\phe- (3-2g(x)(3-2x))[g'(x)(6-24x+15x^2)+g(x)(-24+30x)].
\end{align*}
Now since $g(x)=x^3 [ x(3-2x)]$ and $g'(x)=x^3[12-10x]$ and $x \in (0,1)$ we can
divide $n_2(x)$ by $x^3/6$ finding that it has the same sign as
\begin{align*}
n_3(x)
  &= \frac{1}{3} \bigl(4 x-2(12-10x) \bigr)(3-2x)[5+g(x)(6 -24x+15x^2)]\\
 & \phe-\frac{1}{6}\bigl(3-2g(x)(3-2x)\bigr)\bigl[(12-10x)(6-24x+15x^2)+x(3-2x)(-24+30x)\bigr]\\
 &= -8(1-x) (3-2x)[5+g(x)(6 -24x+15x^2)]\\
 &\phe- (3-2g(x)(3-2x))[-35x^3+93x^2-70x+12]\\
 &= -8 (1-x)(3-2x)[5+g(x)(6 -24x+15x^2)]\\
 &\phe- (3-2g(x)(3-2x)) (1-x)(35x^2-58x+12).
\end{align*}
We can divide $n_3(x)$ by $-(1-x)$ getting a polynomial $n_4(x)$
with the opposite sign from $n_3$.
This yields
\begin{align*}
n_4(x) &= 8(3-2x)[5+g(x)(6-24x+15x^2)]+(3-2g(x)(3-2x))(35x^2-58x+12)\\
      &= 8 (3-2x)[-30x^7+93x^6-84x^5+18x^4+5]\\
 &\phe+  (-8x^6+24x^5-18x^4+3) (35x^2-58x+12) \\
 &= 8 [60x^8-276x^7+447x^6-288x^5+54x^4-10x+15]\\
 &\phe+ (-280x^8+1304 x^7-2118x^6+1332x^5-216x^4+105x^2-174x+36) \\
& = 200 x^8-904x^7+1458x^6-972x^5+216x^4+105x^2-254x+156.
\end{align*}
Note that the coefficient of $x^3$ in $n_4$ is a zero that must not be left out when entering the coefficients into symbolic differentiation codes.

We want to show that $n_4(x)>0$ for $x\in(0,1)$ which
then makes $n_3(x)<0$ and ultimately $f'(x)<0$
so that $\eff_{\ts}'(\delta)>0$ for $\delta=1-x \in (0,1)$.

Our next step is to show that $n_4''(x)>0$ for
all $x \in [0,1]$.  Note that
$$n_4''(x)= 11200x^6-37968x^5+43740x^4-19440x^3+2592x^2+210.$$
Graphing $n_4''(x)$ versus $x$
and evaluating it numerically, makes it is clear that $148<n_4''(x)<369$.
Our next few steps are just to  eliminate even an unreasonable doubt
about the sign of $n_4''$.  Some readers might prefer to skip
that and go instead to the subsection marked ``Conclusion
of the proof''.

\subsection*{Positivity of $\boldsymbol{n''_4}$}
We begin by writing
$$n_4'''(x)=48x(1400x^4- 3955x^3+3645 x^2-1215x+108).$$
Then by the triangle inequality,
\begin{align}\label{eq:triangle}
\vert n_4'''(x) \vert \le 48(1400+3955+3645+1215+108)=495504 \le 2^{19}
\end{align}
for all $x \in [0,1]$.

Now for $k \in \{0,1,\ldots,2^{12} \}$ define $x_k=k 2^{-12}$.
For each $k$
we will let $\widehat{n_4''(x_k)}$ be the numerical evaluation for the polynomial $n_4''(x_k)$ computed using Horner's method with double-precision in R,  implemented with the `horner' function in the \textbf{pracma} R package of \cite{pracma}.

By formula (5.3) in \cite{Higham2002},
the absolute error in Horner's method for
the polynomial $\sum_{r=0}^n a_r x^r$ is at most
$\gamma_{2n} \tilde{p}(\vert x \vert)$ where $\gamma_{2n} \equiv {2nu}/(1-2nu)$,
$u$ is the unit roundoff, and $\tilde{p}(\vert x \vert)= \sum\limits_{r=0}^n \vert a_r \vert \vert x \vert^r$.

Applying that bound to our 6th degree polynomial $n_4''(x)$, and noting that for each $x \in [0,1]$, $\tilde{p}(\vert x \vert)$ will be at most the sum of the absolute value of the coefficients we find that
\begin{align}\label{eq:maxerro}
\max_{0\le k\le 2^{12}}
\big|n_4''(x_k)- \widehat{n_4''(x_k)} \big| \le \gamma_{2 \times 6} \tilde{p}(\vert x_k \vert) \le \frac{12u}{1-12u} \times 115{,}150.
\end{align}
We need not worry about floating point error induced by evaluating $x_k$, because each $x_k$ is a floating point number.
For double precision in R, the unit roundoff is $u=2^{-53} \le 2 \times 10^{-15}$
from which $12u/(1-12u) \le 10^{-13}$,
and so the maximum
error in~\eqref{eq:maxerro} is at most $10^{-7}$.
The smallest value of $\widehat{n_4''(x_k)}$ among all $2^{12}+1$
evaluation points $x_k$ was $148.5743$
and so $\min_{0\le k\le 2^{12}}n_4''(x_k)\ge148$.

Now we are ready to prove that $n_4''(x)>0$ for all $x \in [0,1]$.
For any $x \in [0,1]$ there exists $k_* \in \{0,1,\ldots,2^{12} \}$
with $\vert x- x_{k_*} \vert \le 2^{-13}$.
By~\eqref{eq:triangle} we know that $n_4''$ is Lipschitz continuous
on $[0,1]$ with Lipschitz constant $2^{19}$.
Therefore
\begin{align*}
n_4''(x) \ge n_4''(x_{k_*}) -|n_4''(x_{k_*})-n_4''(x)| \ge 148-2^{19}\times2^{-13}>0,
\end{align*}
holds for all $x\in[0,1].$

\subsection*{Conclusion of the proof}
We have shown above that $n_4''(x)>0$  for any $x \in [0,1]$.
Now since $n_4'(1)=-20$, and $n_4''(x)>0$ for all $x \in [0,1]$, it follows that $n_4'(x) < 0$ for all $x \in [0,1]$. Finally since $n_4(1)=5$ and $n_4'(x)<0$ for all $x \in [0,1]$, it follows that $n_4(x)>0$ for all $x \in [0,1]$.
This completes the proof.  $\qed$

\end{document}